\documentclass[12pt]{l4dc2023}

\usepackage{hyperref}
\hypersetup{
    colorlinks=true,
    linkcolor=blue,
    filecolor=red,      
    urlcolor=blue,
    citecolor=blue,
    linktocpage
}
\usepackage{times}
\usepackage{comment} 

\usepackage{array, makecell}

\usepackage{enumitem}

\usepackage{indentfirst}

\usepackage{bigints}
\usepackage{nccmath} 
\usepackage{xpatch}
\xpatchcmd{\NCC@ignorepar}{%
\abovedisplayskip\abovedisplayshortskip}
{%
\abovedisplayskip\abovedisplayshortskip%
\belowdisplayskip\belowdisplayshortskip}
{}{}

\usepackage{showidx}
\usepackage{tikz}
\usepackage{optidef}
\usepackage{graphicx}
\usepackage{xcolor}
\usepackage{hyperref}
\usepackage[english]{babel}

\usepackage{hyperref}

\usepackage{floatrow, ragged2e, booktabs}
\usepackage{verbatim}
\usepgflibrary{arrows}
\usepackage{verbdef}
\usepackage{color}
\usepackage{graphicx}
\usepackage[european]{circuitikz}
\usepackage{arydshln}
\usepackage[T1]{fontenc}
\usepackage[thinlines]{easytable}

\usepackage{mathtools}
\usepackage[amsmath]{empheq}

\usepackage{doi}

\tikzset{>=latex}

\tikzset{%
block/.style    = {draw, thick, rectangle, minimum height = 3em,
    minimum width = 3em},
  block1/.style    = {draw, thick, rectangle, minimum height = 1.7em,
    minimum width = 1.7em,fill=gray!70},
      block2/.style    = {draw, thick, rectangle, minimum height = 1.7em,
    minimum width = 1.7em,fill=gray!30},
  sum/.style      = {draw, circle, node distance = 1.8cm}, 
}

\firstpageno{1}

\renewcommand{\tilde}{\widetilde}

\renewcommand{\hat}{\widehat}
\newcommand{\FF}{{{\rm I \kern -0.2em R}}}
\newcommand{\RR}{{{\rm I \kern -0.2em R}}}

\newcommand{\CC}{{{\mbox{\rm \hspace*{0.05ex}
\rule[.18ex]{.18ex}{1.24ex} \kern -.65em C}}}} 

\newcommand{\bea}{\begin{eqnarray}}
\newcommand{\eea}{\end{eqnarray}}

\newtheorem{thm}{Theorem}[section]
\newtheorem{rem}{Remark}
\newtheorem{prop}[thm]{Proposition}
\newtheorem{lem}[thm]{Lemma}
\newtheorem{coro}[thm]{Corollary}
\newtheorem{defn}[thm]{Definition}
\newtheorem{assumption}{Assumption}

\newcommand{\ba}{\left[ \begin{array}}
\newcommand{\baa}{\begin{array}}
\newcommand{\ea}{\end{array} \right]}
\newcommand{\eaa}{\end{array}}

\newcommand{\be}{\begin{equation}}
\newcommand{\ee}{\end{equation}}
\newcommand{\bb}{\begin{equation}\label}

\newcommand{\thmref}[1]{Theorem~\ref{#1}}

\newcommand{\lemref}[1]{Lemma~\ref{#1}}

\newcommand{\remref}[1]{Remark~\ref{#1}}
\newcommand{\defref}[1]{Definition~\ref{#1}}
\newcommand{\propref}[1]{Proposition~\ref{#1}}

\newcommand{\figref}[1]{Figure~\ref{#1}}
\def\math#1{\ifmmode{#1} \else {$#1$}\fi}

\newcommand{\sg}{\ifmmode \Sigma \else $\Sigma$ \fi}

\date{}

\begin{document}

\title[Sample Complexity for learning the Robust Observer]{Sample Complexity for Evaluating the Robust Linear Observer's \\  Performance under Coprime Factors Uncertainty}

\author{\Name{Yifei Zhang} \Email{yzhang133@stevens.edu} \\
       \addr Department of Electrical and Computer Engineering\\
       Stevens Institute of Technology,
       NJ 07030, USA
        \AND
       \Name{Sourav Kumar Ukil} \Email{sukil@stevens.edu} \\
       \addr Department of Electrical and Computer Engineering\\
       Stevens Institute of Technology,
       NJ 07030, USA
       \AND
       \Name{Andrei Speril\u{a}}   \Email{andrei.sperila@upb.ro} \\
       \addr Department of Automatic Control and Computer Science\\
       ‘‘Politehnica’’ University of Bucharest, 
       Bucharest, Romania
       \AND
       \Name{{\c{S}}erban Sab\u{a}u}   \Email{ssabau@stevens.edu} \\
       \addr Department of Electrical and Computer Engineering\\
       Stevens Institute of Technology, 
       NJ 07030, USA}


\maketitle

\begin{abstract}
This paper addresses the end-to-end sample complexity bound for learning  in closed loop the state estimator-based robust $\mathcal{H}_2$ controller for an unknown (possibly unstable) Linear Time Invariant ({\bf LTI}) system, when given a fixed state-feedback gain. We build on the results from \cite{Ding1994} to bridge the gap between the parameterization of all state-estimators and the celebrated Youla parameterization. Refitting the expression of the relevant closed loop allows for the optimal linear observer problem given a fixed state feedback gain to be recast as a convex problem in the Youla parameter. The robust synthesis procedure is performed by considering bounded additive model uncertainty on the coprime factors of the plant, such that a min-max optimization problem is formulated for the robust $\mathcal{H}_2$ controller via an observer approach. The closed-loop identification scheme follows \cite{Zheng2022arxiv}, where the nominal model of the true plant is identified by constructing a Hankel-like matrix from a single time-series of noisy, finite length input-output data by using the ordinary least squares algorithm from \cite{Dahleh2020}. Finally, a $\mathcal{H}_{\infty}$ bound on the estimated model error is provided, as the robust synthesis procedure requires bounded additive uncertainty on the coprime factors of the model. 

 
\end{abstract}
$ \fontdimen14\textfont2=6pt
\fontdimen16\textfont2=2.5pt$
\section{Introduction}

State estimation is a fundamental problem in control theory and machine learning. The utilization of state observers has been proven to be significant in both detecting and identifying faults in dynamical systems as well as monitoring and regulating those systems since the work of \cite{Luenberger1966}. The existence of disturbances and uncertainties provides significant difficulties in real-world applications, as practically all observer designs are based on the mathematical model of the plant. For this purpose, a number of sophisticated observer designs have been put out as solutions to the high-performance, robust observer-based regulator design challenge, which has lately attracted significant interest. 
\noindent
The classical LQ control problem for LTI systems, served as the starting point for the aforementioned research problems, where the goal is to identify the best output feedback law that minimizes the expected value of a quadratic cost. In the past few years, significant research has been put into using modern statistical and optimization tools from the machine learning framework to approach classical control problems, see for instance \cite{Dean2018}, \cite{Boczar2018}, \cite{mania2019}, \cite{dean2020}, \cite{furieri2020}, \cite{Matni2015}, \cite{Lee2020}, \cite{matni2020}. 

An end-to-end sample-complexity bound of learning observer-based $\mathcal{H}_2$ controller for an unknown (potentially unstable) LTI plant that stabilizes the true system with high probability is established in this paper by incorporating recent advances in finite-time system identification. The resulting sub-optimal gap is bounded as a function of the level of model uncertainty. The  end-to-end sample complexity bound for learning the robust observer-based $\mathcal{H}_2$ controller is $\mathcal{O}\Bigg( \dfrac{\sqrt{\dfrac{logT}{T}}}{1 - \alpha \sqrt{\dfrac{logT}{T}}} \Bigg)$, where $T$ is the time horizon for learning. 

\textbf{Paper Organization}: The paper is organized as follows: the general setup and problem formulation is given in \hyperref[generalsetupandpreliminaries]{Section II}. The robust observer synthesis with uncertainty on the coprime factors is included in \hyperref[robustcontrollersynthesis]{Section III}. A brief discussion on the sub-optimality guarantees with end-to-end sample complexity results are stated in \hyperref[endtoendanalysis]{Section IV}. Conclusion and future 
directions are given in \hyperref[conclusion]{Section V}. 
All the proofs are postponed to the \hyperref[appendix]{Appendices}, where literature review, mathematical preliminaries and closed loop system identification scheme also have been discussed briefly.

\useshortskip

\section{General Setup and Technical Preliminaries}\label{generalsetupandpreliminaries}
The notation used in this paper is fairly common in control systems. 
Upper and lower case boldface letters 
(e.g. ${\bf G}$) are used to denote transfer function matrices, while lower and upper case letters (e.g. $z$ and $A$) are used to denote vectors and matrices. The enclosed results are valid for  discrete-time linear systems, therefore $z$ denotes the complex variable associated  with the $\mathbf{Z}$-transform for discrete-time systems. 
A LTI system is {\em stable} if all the poles of its TFM are situated 
inside the unit circle for discrete time systems. The TFM of a LTI system is called {\em unimodular} if it is square, stable and has a stable inverse. For the sake of brevity the $z$ argument after a transfer function may be omitted. 
$\mathbb{R}(z) $ denotes the set of all real--rational  transfer functions and $\mathbb{R}(z)^{p \times m} $ denotes the set of $p \times m$ matrices having all entries in $\mathbb{R} (z)$.
The notation  ${\bf T}^{ \ell \varepsilon}$ is used to indicate  the
 mapping  from signal $\varepsilon$ to signal $\ell$  after combining 
all the ways in which $\ell$ is a function of $\varepsilon$ and solving any
feedback loops that may exist. 
For example, ${\bf T}^{zw}$ is the mapping from the disturbances $w$ to the regulated measurements $z$.

\subsection{The State Estimation Problem}

For a discrete-time {\bf  LTI} (Linear and Time Invariant) systems driven by Gaussian process and sensor noise, the state-space model is given by:
\useshortskip
\begin{equation}\label{stateeq}
\begin{aligned}
      x_{k+1} & = {A} x_k + {B} (u_k + w_k) + \delta_k,\\  y_k & = {C} x_k + {D} u_k + \nu _k,
\end{aligned}
\end{equation}
\noindent
where $x_k \in \mathbb{R}^n$ is the state of the system, $u_k \in \mathbb{R}^m$ is the control input and $y_k \in \mathbb{R}^p$ is the measured output and $w_k \in \mathbb{R}^m$, $\delta_k \in \mathbb{R}^n$ are the control additive and state additive disturbances, while $\nu_k \in \mathbb{R}^p$ is the measurement noise, all considered to be Gaussian  with zero mean and  covariance matrices $\sigma_{w}^2 I$, $\sigma_{\delta}^2 I$ and $\sigma_\nu^2 I$ 
respectively. 

A state estimator (observer) for \eqref{stateeq} is defined as a system that provides an estimate $\hat{x}_k$ of the internal state $x_k$, while having access solely to the control input $u$ and measured output $y$, with the underlying requirement that the estimation error converges to zero in the steady-state, that is $ \displaystyle \lim_{k \rightarrow \infty} (x_k - \hat{x}_k) = 0$.
A state estimator is generically of the form
\begin{equation} \label{ses1}
    \hat{x}(z) = {{\bf \Psi}^{u}}(z) u(z) + {{\bf \Psi}^{y}}(z) y(z)
\end{equation}
where ${\bf \Psi}^{u}(z)$ and ${\bf \Psi}^{y}(z)$ are two LTI filters (stable Transfer Function Matrices ({\bf TFM}s)) {\em for the design of which one needs to know the model \eqref{stateeq} of the plant}, see for example \cite{Ding1994}. The celebrated Kalman Filter, represents the canonical formulation of performance specifications for a state estimator \eqref{ses1} as it  minimizes the transfer from the exogenous signals in \eqref{stateeq} (e.g. the measurement noise $\nu_k$) to the estimation error $x_k - \hat{x}_k$ (by using for example norm based costs).

\subsection {Output Feedback Stabilizing Controllers}
A standard unity feedback configuration is depicted in \figref{2Block} , where ${\bf G}\in \mathbb{R}(z)^{p \times m}$ is a multi-variable LTI plant and ${\bf K}\in \mathbb{R}(z)^{m \times p}$ is an LTI controller. Here $w$, $\nu$ and $r$  are the input disturbance, sensor noise and reference signal respectively while $u$, $z$ and $y$ are the controls, regulated signals  and measurements vectors, respectively.  

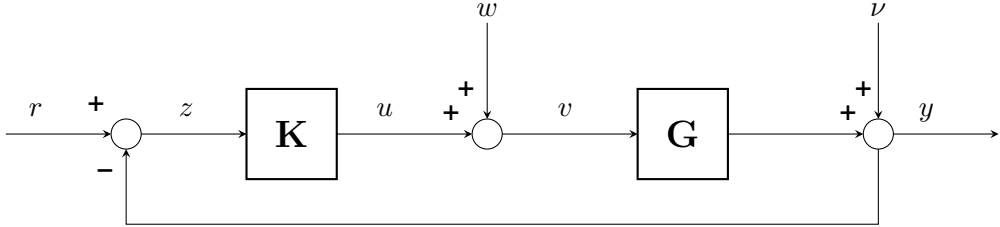
\begin{figure}[ht]
\begin{tikzpicture}[scale=0.4]
\draw[xshift=0.1cm, >=stealth ] [->] (0,0) -- (3.5,0);
\draw[ xshift=0.1cm ]  (4,0) circle(0.5);
\draw[xshift=0.1cm] (3,1)   node {\bf{+}} (1,0.8) node {$r$};
\draw [xshift=0.1cm](6,0.8)   node {$z$} ;
\draw[ xshift=0.1cm,  >=stealth] [->] (4.5,0) -- (8,0);
\draw[ thick, xshift=0.1cm]  (8,-1.5) rectangle +(3,3);
\draw [xshift=0.1cm](9.5,0)   node {\Large{${\bf K}$}} ;
\draw[ xshift=0.1cm,  >=stealth] [->] (11,0) -- (15.5,0);
\draw[ xshift=0.1cm ]  (16,0) circle(0.5cm);
\draw [xshift=0.1cm](12.6,0.8)   node {$u$} ;
\draw [xshift=0.1cm](18.6,0.8)   node {$v$} ;
\draw [xshift=0.1cm] (14.8,0.7)   node {\bf{+}};
\draw[  xshift=0.1cm,  >=stealth] [->] (16,3.7) -- (16,0.5);
\draw [xshift=0.1cm] (16,3.6)  node[anchor=south] {$w$}  (15.3,1.5)  node {\bf{+}};
\draw[  xshift=0.1cm,  >=stealth] [->] (16.5,0) -- (21,0);
\draw[ thick, xshift=0.1cm ]  (21,-1.5) rectangle +(3,3) ;
\draw [xshift=0.1cm] (22.5,0)   node {\Large{${\bf G}$}} ;
\draw[ xshift=0.1cm,  >=stealth] [->] (24,0) -- (28.5,0);
\draw[ xshift=0.1cm ] (29,0)  circle(0.5);
\draw [xshift=0.1cm] (29,3.6)  node[anchor=south] {$\nu$}  (28.5,1.5)  node {\bf{+}};
\draw [xshift=0.1cm] (28,0.7)   node {\bf{+}};
\draw [xshift=0.1cm] (30.6,0.7)   node {$y$};
\draw[  xshift=0.1cm,  >=stealth] [->] (29.5,0) -- (33,0);
\draw[  xshift=0.1cm,  >=stealth] [->] (29,3.7) -- (29,0.5);
\draw[ xshift=0.1cm,  >=stealth] [->] (29,-0.5) -- (29,-3) -- (4,-3)-- (4, -0.5);
\draw [xshift=0.1cm] (3.3,-1.3)   node {\bf{--}};
\useasboundingbox (0,0.1);
\end{tikzpicture}
\caption{Standard unity feedback loop of the plant $\bf G$ with the controller $\bf K$}
\label{2Block}
\end{figure}

If all the closed--loop maps  from the exogenous signals $\displaystyle [ r^T\; \: w ^T \; \: \nu^T \;]^T$ to 
any point inside the feedback loop 
are stable, then ${\bf K}$ is said to be an (internally) stabilizing controller of ${\bf G}$ or equivalently  that ${\bf K}$  stabilizes ${\bf G}$.

\subsection{The Youla-Ku\c{c}era Parameterization of All Stabilizing Controllers}\label{2}

\begin{defn} 
(\cite{vidyasagar1985}) \label{DCFdefinition}
A collection of eight  stable TFMs $\big({\bf M}, {\bf  N}$, $\tilde {\bf  M}, \tilde {\bf  N}$, ${\bf X}, {\bf  Y}$, $\tilde {\bf  X}, \tilde {\bf  Y}\big)$ is called a  {\em Doubly Coprime Factorization} ({\bf DCF}) of the plant ${\bf  G}$  if   $\tilde {\bf  M}$ and ${\bf M}$ are invertible, yield the coprime
factorizations
${\bf G}=\tilde{\bf M}^{-1}\tilde{\bf N}={\bf NM}^{-1}$, and satisfy the following equality  (B\'{e}zout's identity):
\begin{equation}\label{bezoutidentity}
\ba{cc}   \tilde {\bf M} &  \tilde {\bf N} \\ - {\bf X} &  {\bf Y} \ea
\ba{cc}  \tilde {\bf Y} & -{\bf N} \\   \tilde {\bf X} &  {\bf M} \ea = I_{p+m},
\ba{cc}  \tilde {\bf Y} & -{\bf N} \\   \tilde {\bf X} &  {\bf M} \ea 
\ba{cc}   \tilde {\bf M} &  \tilde {\bf N} \\ - {\bf X} &  {\bf Y} \ea
 =  I_{p+m}.
\end{equation}
\end{defn}

\begin{thm} (\cite{Ding1994}, \cite{vidyasagar1985}) \label{dingTFM}
Given a stabilizable and detectable state-space realization \eqref{stateeq} of the plant ${\bf G}$, then a DCF as in \defref{DCFdefinition} above is given by:
 \begin{equation} \label{DingTFMequations}
    \begin{split}
        \textbf{M}(z) = {I} + {F}(z{I}-{A}_F)^{-1} {B}, \hspace{5pt}  \textbf{N}(z) = {C_F}(z{I}-{A}_F)^{-1}{B}\\
       \tilde{\textbf{M}}(z) = {I} - {C}(z{I}-{A}_L)^{-1} {L}, \hspace{5pt}  \tilde{\textbf{N}}(z) = {C}(z{I}-{A}_L)^{-1}{B_L}\\
        \textbf{X}(z) = -{F}(z{I}-{A}_L)^{-1}{L}, \hspace{5pt}  \textbf{Y}(z) = {I} - {F}(z{I}-{A}_L)^{-1} {B_L}\\
        \tilde{\textbf{X}}(z) = -{F}(z{I}-{A}_F)^{-1}{L}, \hspace{5pt}  {\bf{\tilde{Y}}}(z) = {I} + {C_F} (z{I}-{A}_F)^{-1} {L}
    \end{split}
\end{equation}
where ${A}_F \overset{def}{=} {A} + {B}{F}$, ${A}_L \overset{def}{=} {A} - {L}{C}$, ${C}_F \overset{def}{=} {C} + {D}{F}$ and ${B}_L \overset{def}{=} {B} - {L}{D} $, where $F$ and $L$ are  stabilizing state-feedback and estimation gains that allocate all eigenvalues of $A_F$ and $A_L$  inside the unit disk. 
\end{thm}

\begin{rem} \label{Miercuri} Theorem~\ref{dingTFM} above  states that the DCF \eqref{DingTFMequations} of the plant is essentially equivalent with establishing certain stabilizing state- feedback $F$ and  estimation gain $L$, such that $u_k = F \hat{x}_k$ in tandem with $\hat{x}_{k+1} = A \hat{x}_k + B u_k + L(y_k - C \hat{x}_k)$ is the output stabilzing controller ${\bf K}={\bf Y}^{-1}{\bf X}$.
\end{rem}



\begin{thm} \label{Youlaaa}
{\bf (Youla-Ku\u{c}era)} \cite[Ch.5]{vidyasagar1985} 
  Let  $\big({\bf M}, {\bf  N}$, $\tilde {\bf  M}, \tilde {\bf  N}$, ${\bf X}, {\bf  Y}$, $\tilde {\bf  X}, \tilde {\bf  Y}\big)$ be a doubly coprime factorization of ${\bf G}$. Any controller ${\bf K}_{\bf Q}$ stabilizing the plant ${\bf G}$, can be written as
\begin{equation}
\label{YoulaEq}
{\bf K_Q}={\bf Y}_{\bf Q}^{-1}{\bf X_Q} = {\bf \tilde{X}}_{\bf Q}  {\bf \tilde{Y}}_{\bf Q}^{-1},
 \end{equation}
where
${\bf X_Q}$, ${\bf \tilde{X}_Q}$, ${\bf Y_Q}$ and ${\bf \tilde{Y}_Q}$ are defined as:  ${\bf X_Q}  \overset{def}{=}  {\bf X}+{\bf Q} \tilde{\bf M}, {\bf \tilde{X}_Q}  \overset{def}{=}  \tilde{\bf X}+{\bf MQ}, {\bf Y_Q}  \overset{def}{=}  {\bf Y} - {\bf Q} \tilde{\bf N},$ and ${\bf \tilde{Y}_Q}  \overset{def}{=}  \tilde{\bf Y}-{\bf NQ}$,
for some stable ${\bf Q}$ in $\mathbb{R}(z)^{m\times p}$. It also holds that ${\bf K_Q}$ from (\ref{YoulaEq}) stabilizes ${\bf G}$, for any stable ${\bf Q}$. 
\end{thm}

\subsection{Parameterization of All State Estimators} \label{Yifei}

The following results provides the parameterization of {\em all} state observers of a given LTI system.

\begin{thm} (\cite{Ding1994}) \label{dingtheorem}
Given stabilizing state-feedback $F$ and  estimation gain $L$, or equivalently, given a DCF \eqref{DingTFMequations} of the LTI plant \eqref{stateeq} (see also Remark~\ref{Miercuri}), let us denote  ${\textbf{P}}(z) \overset{def}{=} (zI - A_F)^{-1}B$. Then:  {\bf (A)} the pair of filters $({{\bf \Psi}^{u}_{}}; {{\bf \Psi}^{y}_{}})$ generate a state estimator 
 \eqref{ses1} for the system in \eqref{stateeq}
if and only if
\begin{equation} \label{observerparaneterizationwoS}
    {{\bf \Psi}^{u}}(z) {\textbf M}(z) + {{\bf \Psi}^{y}}(z) {\textbf N}(z) = {\bf P}(z).
\end{equation}

\noindent {\bf (B)} Furthermore, {\em any} state estimator for \eqref{stateeq} can be written as 
\begin{equation} \label{STES}
\hat{x}(z) = {{\bf \Psi}^{u}_{\bf S}}(z) u(z) + {{\bf \Psi}^{y}_{\bf S}}(z) y(z),
\end{equation}
where
\begin{equation} \label{Psi1Psi2}
        {{\bf \Psi}^{u}_{\bf S}}(z) \overset{def}{=} {\textbf P}(z) {\textbf Y}(z) + {\textbf S}(z) \tilde{\textbf{N}}(z), \quad
         {{\bf \Psi}^{y}_{\bf S}}(z) \overset{def}{=} {\textbf P}(z) {\textbf X}(z) - {\textbf S}(z) \tilde{\textbf{M}}(z) 
\end{equation}
for some stable   ${\textbf S}(z) \in \mathbb{R}(z)^{n\times p}$. Conversely, for any stable  ${\textbf S}(z)$ it holds that \eqref{STES}, with $({{\bf \Psi}^{u}_{\bf S}},{{\bf \Psi}^{y}_{\bf S}})$ as in \eqref{Psi1Psi2}, is a state estimator for \eqref{stateeq}.
\end{thm}

\begin{rem}   The intrinsic connections of Theorem~\ref{dingtheorem} with output feedback stabilization are evident,  just as the affine parameterization \eqref{Psi1Psi2} of all state-estimators is akin to the Youla parameterization of Theorem~\ref{Youlaaa}, but  it is important to note that  Theorem~\ref{dingtheorem} holds just the same if the plant \eqref{stateeq} is in open loop or if the plant is in a feedback interconnection with a stabilizing controller ${\bf K}$. However,  these two very distinct scenarios must be handled with care. In this paper we are interested in "learning"  the optimal state estimator of an unknown plant {\em in closed feedback loop}. To this end the  following two results (for the closed-loop scenario) will be instrumental towards the main result and surprisingly enough, they cannot be found in the original work from \cite{Ding1994}.
\end{rem}
\begin{thm} \label{TransferfromNoisetoError}
Consider the the LTI plant \eqref{stateeq} in feedback interconnection with the controller given by: $u_k = F \hat{x}_k$ in tandem with any state-estimator of the form $\hat{x} = {{\bf \Psi}^{u}_{\bf S}}(z) u(z) + {{\bf \Psi}^{y}_{\bf S}}(z) y(z)$. The closed loop maps from the disturbances $w$ and measurement noise $\nu$ to the estimation error $(x-\hat{x})$ are affine functions of the ${\bf S}$ parameter from Theorem~\ref{dingtheorem} {\bf (B)}, moreover:
\begin{equation} \label{claps}
   T^{(x-\hat{x})w}_{\bf S} = {{\bf \Psi}^{u}_{\bf S}}(z) \; \; \text{and} \; \; T^{(x-\hat{x})\nu}_{\bf S} = - {{\bf \Psi}^{y}_{\bf S}}(z), \;   \text{respectively}.
\end{equation}
\end{thm}

\begin{thm} \label{QS}
 Given a DCF \eqref{DingTFMequations} of the LTI plant \eqref{stateeq} and its subsequent  stabilizing state-feedback gain $u_k = F \hat{x}_k$, let us assume that the $F\in\mathbb{R}^{m \times n}$ matrix is left invertible (it has full column rank). Then any stabilizing output feedback controller ${\textbf{K}}_{\textbf{Q}}$ from \eqref{YoulaEq}  can be realized as: $u_k = F \hat{x}_k$ in tandem with the state-estimator $\hat{x} = {{\bf \Psi}^{u}_{\bf S}}(z) u(z) + {{\bf \Psi}^{y}_{\bf S}}(z) y(z)$ from \eqref{STES}, where
\begin{equation} \label{cheia}
    {\text F} \textbf{S}(z) = \textbf{Q}(z) + \tilde{\textbf{X}}(z)
\end{equation}
\end{thm}
\begin{rem} \label{Fdesign}
The two theorems above clarify the fact that the two filters that realize any state estimator \eqref{STES} {\em in closed-loop} are actually the closed loop maps from the exogenous signals to the estimation error. Furthermore, and this is important, {\em under the assumption that the state-feedback gain matrix $F$ is left invertible}, there exists a bi-univocal relationship \eqref{cheia} between the  ${\bf S}$ parameter from Theorem~\ref{dingtheorem} and the Youla parameter. This immediately allows  to rephrase  parameterization \eqref{Psi1Psi2}  of all state observers which is affine in ${\bf S}$, to a parameterization $({{\bf \Psi}^{u}_{\bf Q}}; {{\bf \Psi}^{y}_{\bf Q}})$ affine in the Youla parameter, thus bridging the gap between any stabilizing controller ${\textbf{K}}_{\textbf{Q}}$ from \eqref{YoulaEq} and its realization via: a {\em fixed}  state-feedback gain $F$ in tandem with the dynamic state estimator $({{\bf \Psi}^{u}_{\bf Q}}; {{\bf \Psi}^{y}_{\bf Q}})$. The fixed state-feedback gain $F$ comes from the initial stabilizing controller in the closed loop, as $\tilde{\textbf{X}}(z) = -{F}(z{I}-{A}_F)^{-1}{L}$ is neither a function of $\bf Q$ nor a function of $\bf S$.
\end{rem}

\subsection {A First Glimpse into the Separation Principle}

We illustrate below the fact that any stabilizing controller can be realized either as a fixed state-feedback gain in tandem with a dynamic state-estimator or as a fixed estimation gain in tandem with dynamic state feedback. Both parameterizations are affine in the Youla parameter.
\vspace{-5pt}
\begin{table}[h!]
  \centering
\begin{tabular}{|m{7.25cm}|m{7cm}|}
\hline
\thead{\cite{Ding1994} and Subsection~\ref{Yifei}} & \thead{\cite{Apkarian1999}}
\\
\hline
\em {Any stabilizing ${\textbf{K}}_{\textbf{Q}}$ from \eqref{YoulaEq} can be realized via the static state-feedback gain $F$ in tandem with the dynamic state estimator $({{\bf \Psi}^{u}_{\bf Q}}; {{\bf \Psi}^{y}_{\bf Q}})$.} & \em{Any stabilizing ${\textbf{K}}_{\textbf{Q}}$ from \eqref{YoulaEq} can be realized via the static estimation  gain $L$ in tandem with the dynamic feedback ${\bf Q}$.}  \\ 
\hline
\thead{$u_k = F \hat{x}_k$} \newline \thead{ $\hat{x} = {{\bf \Psi}^{u}_{\bf Q}} u + {{\bf \Psi}^{y}_{\bf Q}}y $ } & \thead{ $\hat{x}_{k+1} = A \hat{x}_k + B u_k + L(y_k - C \hat{x}_k)$} \newline \thead{ $u = F \hat{x} + \textbf{Q} (y - C \hat{x})$ } \\
\hline
\end{tabular}
\end{table}
\vspace{-5pt}
\subsection{The Optimal State Estimator}

\noindent\fbox{
    \parbox{0.97\textwidth}{THE PROBLEM: In this paper we consider the {\em unknown} plant \eqref{stateeq} in feedback interconnection with some {\em known} stabilizing controller ${\bf K}$, controller that is realized as: a fixed state-feedback gain $F$ considered to be immutable, namely $u_k = F \hat{x}_k$, in tandem with some state-estimator $\hat{x} = {\bf \Psi}^{u} u + {\bf \Psi}^{y} y$. First we must learn the unknown system with high probability, in finite time, from a single trajectory in the closed loop. Finally, we must design the optimal state-observer that in tandem with the state-feedback gain $u_k = F \hat{x}_k$ yields the optimal LQG performance.
}
}

The canonical formulation of performance specifications
for a state estimator is to minimize the transfer from the exogenous signals in  to the estimation error. However, as stated above, the declared scope is to design a state-estimator specifically tailored to work in tandem with  the fixed state-feedback gain  $u_k = F \hat{x}_k$. In this context, the choice of the optimality criterion is essential, as outlined below.

\begin{prop} \label{optimalcontrolproblem}
We define the Optimal Observer Evaluation Problem, given a fixed state-feedback gain $\text{F}$ with $u = F \hat{x}$  as:
\begin{equation}
\label{criterion}    \min_{\bf Q\: \text{stable}}  \big\| \ba{cc}  F {\bf \Psi}^{u}_{\bf Q} & -F{\bf \Psi}^{y}_{\bf Q} \ea \big\|_{\mathcal{H}_2}
\end{equation}

which turns out to be equivalent with:
\begin{equation} \label{YifeiZ}
\centering
    \min_{\bf Q\: \text{stable}}  \Big\| \ba{@{}c@{\kern 0.4em}c@{}} {I_m - {\bf Y_Q}(z) + \big(I_m - {\bf M}(z)\big) {\bf Q}(z) \tilde{\bf N}(z)} \; \; & \; \; {{\bf X_Q}(z) + \big(I_m - {\bf M}(z)\big) {\bf Q}(z) \tilde{\bf M}(z)} \ea \Big\|_{\mathcal{H}_2}
\end{equation}

\end{prop}

\noindent
Proof of \thmref{dingtheorem}, \thmref{TransferfromNoisetoError}, \thmref{QS} and \propref{optimalcontrolproblem} is provided in \hyperref[appendixB]{Appendix B}.


\begin{rem} (Estimation Error)  \label{FdesignRemark}
The reason behind choosing \eqref{criterion} for the observer design in this context is mainly caused by the fact that the model of the plant can never be determined with absolute accuracy, since any learning algorithm produces outcomes which are inherently uncertain.  Furthermore, the objective function from\eqref{criterion}  pertains to the difference in $\mathcal{H}_2$ performance in the closed loop between the state-feedback control $u=Fx$ (with direct access to the state) and any output feedback controller  ${\bf K_Q}$. The thorough reasoning for this and all the underlying implications are deferred to \hyperref[appendixC]{Appendix C}. 


\end{rem}

\section{Robust Controller Synthesis: An Observer Based Approach}\label{robustcontrollersynthesis}

\noindent\fbox{
    \parbox{0.97\textwidth}{
    The outcome of the "learning" of the true plant ${\bf G}$ from closed-loop measurements comes in the form of a left coprime factorization of what we have dubbed {\em the nominal model}\footnote{Or perhaps, just as suited {\em the learned model}}, namely ${\bf G}^{\mathtt{md}} = (\tilde{\bf M}^{\mathtt{md}})^{-1} \tilde{\bf N}^{\mathtt{md}} = {\bf N}^{\mathtt{md}} ({\bf M}^{\mathtt{md}})^{-1}$. For the detailed description of the learning algorithm we refer to the Appendix~G from  \cite{Zheng2022arxiv}. In order to evaluate the  discrepancy between the learned ${\bf G}^{\mathtt{md}}$ and the true plant, we make a recourse to the preeminent method for modelling uncertainty for LTI systems (stemming from classical robust control), specifically via additive perturbations on the coprime factors.  
}
}

\begin{rem} On top of being able to cope with learning unstable plants (in closed loop), this method of modelling uncertainty, explicitly avoids the need of knowing {\em apriori} the McMillan degree ({\em i.e.} the state dimension of a minimal state-space realization) of the unknown plant, which is never known in practice. The flip side of this coin, is that since the learned nominal model ${\bf G}^{\mathtt{md}}$ and the true plant will not even have the same McMillan degree, it is impossible to retrieve anything about the state representation \eqref{stateeq} of the true plant solely from the knowledge of ${\bf G}^{\mathtt{md}}$. Consequently designing  an optimal, robust state-estimator for the true plant only on the basis of learned nominal model appears to be a daunting task.
\end{rem}



\noindent
With the DCF of the  nominal model of the plant ${\bf G}^{\mathtt{md}} = (\tilde{\bf M}^{\mathtt{md}})^{-1} \tilde{\bf N}^{\mathtt{md}} = {\bf N}^{\mathtt{md}} ({\bf M}^{\mathtt{md}})^{-1}$, we can write the B\'{e}zout's identity that incorporates 
the coprime factorization of the initial, known stabilizing controller\footnote{The controller with which the closed-loop learning is being performed is assumed to be known.} ${{\bf K}^{\mathtt{md}}} = ({\bf Y}^{\mathtt{md}})^{-1} {\bf X}^{\mathtt{md}} = {\bf {\tilde{X}^{\mathtt{md}}}} ({\bf \tilde{Y}^{\mathtt{md}}})^{-1}$, specifically:

\begin{equation}\label{bezout2}
\ba{cc}   {\tilde {\bf M}}^{\mathtt{md}} &  {\tilde {\bf N}}^{\mathtt{md}} \\ - {\bf X}^{\mathtt{md}} &  {\bf Y}^{\mathtt{md}} \ea
\ba{cc}  {\bf \tilde{Y}^{\mathtt{md}}} & -{{\bf N}}^{\mathtt{md}} \\   {\bf {\tilde{X}^{\mathtt{md}}}} &  {{\bf M}}^{\mathtt{md}} \ea = \ba{cc}   { I_p} &   {0} \\  {0} &  { I_m} \ea.
\end{equation} 

\begin{defn}[Model Uncertainty Set] \label{modeluncertaintyset}
The $\gamma$-radius {\em model uncertainty} set for the nominal plant ${\bf G}^{\mathtt{md}}$ \textrm{with} \;
$\Delta_{\bf \tilde{M}}$, $\Delta_{\bf \tilde{N}}$ {both stable} is defined as:
\begin{equation} \label{Youlamd}
    \mathcal{G}_\gamma \overset{\mathrm{def}}{=} \{ {\bf G} = {\tilde {\bf M}}^{-1} {\tilde {\bf N}} \hspace{2pt} \big | \hspace{2pt} {\tilde {\bf M}} = ({\tilde {\bf M}^{\mathtt{md}}}+\Delta_{\bf {\tilde{M}}}),  {\tilde {\bf N}} = ({\tilde {\bf N}}^{\mathtt{md}}+\Delta_{\bf \tilde{N}});  \; \;  \; \Big\| \ba{cc}  \Delta_{\bf \tilde{M}} &  \Delta_{\bf \tilde{N}} \ea \Big\|_{\infty}  < \gamma \}
\end{equation}
\end{defn}

\begin{defn}[$\gamma$-Robustly Stabilizing]\label{perturbedplant} A fixed stabilizing controller ${\bf K}$ of the nominal plant  
  is said to be $\gamma$-robustly stabilizing iff ${\bf K}$ stabilizes not only ${\bf G}^{\mathtt{md}}$ but also all plants  ${\bf G} \in \mathcal{G}_\gamma$. 
\end{defn}

\begin{assumption} \label{True plant}
 It is assumed that the  true plant, denoted by ${\bf G}^{\mathtt{pt}}$, belongs to the   model uncertainty set introduced in \defref{modeluncertaintyset}, i.e. that there exist stable $\Delta_{\bf \tilde{M}}$, $\Delta_{\bf \tilde{N}}$ with  $ \Big\| \ba{cc}  \Delta_{\bf \tilde{M}} &  \Delta_{\bf \tilde{N}} \ea \Big\|_{\infty}  < \gamma$ for which ${\bf G}^{\mathtt{pt}} = ({\tilde {\bf M}}^{\mathtt{md}}+\Delta_{\bf \tilde{M}})^{-1}  ({\tilde {\bf N}}^{\mathtt{md}}+\Delta_{\bf \tilde{N}})$.
\end{assumption}

\noindent In the presence of additive uncertainty on the coprime factors the B\'{e}zout's identity in \eqref{bezout2} no longer holds, however, the following holds for {\em certain stable}  ${\Delta_{\bf M}}$, ${\Delta_{\bf N}}$ {\em factors}:
\begin{equation}\label{PhiMatrix}
\ba{cc}    ({\tilde {\bf M}}^{\mathtt{md}}+\Delta_{\bf \tilde{M}})  &   ({\tilde {\bf N}}^{\mathtt{md}}+\Delta_{\bf \tilde{N}}) \\ - {{\bf X}_{\bf Q}^{\mathtt{md}}}  &  {{\bf Y}_{\bf Q}^{\mathtt{md}}}  \ea
\ba{cc}  {\tilde {\bf Y}_{\bf Q}^{\mathtt{md}}} & -({{\bf N}}^{\mathtt{md}}+\Delta_{\bf {N}}) \\   {\tilde {\bf X}_{\bf Q}^{\mathtt{md}}} &  ({{\bf M}}^{\mathtt{md}}+\Delta_{\bf {M}}) \ea = \ba{cc}    {\bf \Phi}_{11}({\bf Q})  &   O \\  O &  {\bf \Phi}_{22}({\bf Q}) \ea.
\end{equation} 
The block diagonal structure of the right hand side term in \eqref{PhiMatrix} is due to the fact that ${\bf G}^{\mathtt{pt}} = ({\tilde {\bf M}}^{\mathtt{md}}+\Delta_{\bf \tilde{M}})^{-1}  ({\tilde {\bf N}}^{\mathtt{md}}+\Delta_{\bf \tilde{N}}) =({{\bf N}}^{\mathtt{md}}+\Delta_{\bf {N}}) ({{\bf M}}^{\mathtt{md}}+\Delta_{\bf {M}})^{-1} $ for  the stable  ${\Delta_{\bf M}}$, ${\Delta_{\bf N}}$  factors from Assumption~\ref{True plant}.

\begin{lem}\label{phi11phi22}
 A stabilizing controller of the nominal plant ${{\bf K}_{\bf Q}^{\mathtt{md}}} = ({\bf Y}_{\bf Q}^{\mathtt{md}})^{-1} {\bf X}_{\bf Q}^{\mathtt{md}} = {\bf {\tilde{X}_{\bf Q}^{\mathtt{md}}}} ({\bf \tilde{Y}_{\bf Q}^{\mathtt{md}}})^{-1}$ is  $\gamma$-robustly stabilizing iff  for any stable model perturbations  $\Delta_{\bf \tilde{M}}, \Delta_{\bf \tilde{N}}$ with  $ \Big\| \ba{@{}cc@{}}  \Delta_{\bf \tilde{M}} &  \Delta_{\bf \tilde{N}} \ea \Big\|_{\infty}  < \gamma$ the TFM 
  \begin{equation} \label{August10_7am}
         {\bf \Phi}_{11}({\bf Q}) = { I_p} + \ba{cc}  \Delta_{\bf \tilde{M}} &  \Delta_{\bf \tilde{N}} \ea \ba{c}  {\tilde {\bf Y}_{\bf Q}^{\mathtt{md}}} \\  {\tilde {\bf X}_{\bf Q}^{\mathtt{md}}} \ea,
 \end{equation}
 from \eqref{PhiMatrix}  is unimodular {\em i.e.}  it is  square, stable and has a  stable inverse. 
 
\end{lem}


\begin{thm}\label{oarasthm1} The Youla parameterization yields a $\gamma$-robustly stabilizing controller ${\bf K_Q}$ iff its corresponding right coprime factors satisfy 
$\Bigg\| \ba{c} {\tilde {\bf Y}}_{\bf Q} \\  {\tilde {\bf X}}_{\bf Q} \ea \Bigg\|_{\infty}  \leq \dfrac{1}{\gamma}$, where ${\bf Q}$ denotes as the Youla parameter. 
\end{thm}


\noindent
The proofs for \lemref{phi11phi22} and \thmref{oarasthm1} are given on \hyperref[appendixB]{Appendix B}.
As an intermediary result, by employing \thmref{oarasthm1} and the standard inequality from \hyperref[appendixA]{Appendix A} it is concluded that:
\begin{center}
$\Bigg\| \ba{cc}  -\Delta_{\bf \tilde{M}} &  -\Delta_{\bf \tilde{N}} \ea
\ba{c} {\tilde {\bf Y}_{\bf Q}} \\  {\tilde {\bf X}_{\bf Q}} \ea \Bigg\|_{\infty}\leq \Bigg\| \ba{cc}  \Delta_{\bf \tilde{M}} &  \Delta_{\bf \tilde{N}} \ea \Bigg\|_{\infty} \Bigg\|
\ba{c}  {\tilde {\bf Y}_{\bf Q}} \\  {\tilde {\bf X}_{\bf Q}} \ea \Bigg\|_{\infty} < \gamma \times \dfrac{1}{
\gamma} = 1$.
\end{center}

\noindent
Starting from the left coprime factorization of the true plant, known to be of the form ${\bf G}^{\mathtt{pt}} = ({\tilde {\bf M}}^{\mathtt{md}}+\Delta_{\bf \tilde{M}})^{-1}  ({\tilde {\bf N}}^{\mathtt{md}}+\Delta_{\bf \tilde{N}})$, 
one can always obtain a DCF of the true plant by redefining  $\tilde{\bf M}^{\mathtt{pt}} \overset{def}{=} {\bf \Phi}_{11}^{-1} \big({\tilde {\bf M}}^{\mathtt{md}}+\Delta_{\bf \tilde{M}}\big),  \tilde{\bf N}^{\mathtt{pt}} \overset{def}{=}  {\bf \Phi}_{11}^{-1} \big({\tilde {\bf N}}^{\mathtt{md}}+\Delta_{\bf \tilde{N}}\big), {\bf M}^{\mathtt{pt}} \overset{def}{=}   \big({ {\bf M}}^{\mathtt{md}}+\Delta_{\bf {M}}\big) {\bf \Phi}_{22}^{-1}$, and ${\bf N}^{\mathtt{pt}} \overset{def}{=}  \big({ {\bf N}}^{\mathtt{md}}+\Delta_{\bf {N}}\big)  {\bf \Phi}_{22}^{-1}$, such that the B\'{e}zout identity holds with the $\bf X_{\bf Q}^{\mathtt{md}}$, $\bf Y_{\bf Q}^{\mathtt{md}}$, $\tilde{\bf X}_{\bf Q}^{\mathtt{md}}$ and $\tilde{\bf Y}_{\bf Q}^{\mathtt{md}}$ factors available from the known controller. Here,  ${\bf \Phi}_{11}$, ${\bf \Phi}_{22}$ are as in \eqref{PhiMatrix}. By re-establishing the B\'{e}zout identity we are able to formulate the robust version of \eqref{YifeiZ} as:

\begin{thm}\label{theoremAcost}
The Robust Linear Observer Evaluation Problem given a fixed state feedback gain $F$ reads :
\begin{equation} \label{RobustcostA}
\begin{aligned}
    \min_{\bf Q\: \text{stable}} \max_{\Big\| \ba{cc}  \Delta_{\bf \tilde{M}} &    \Delta_{\bf \tilde{N}}  \ea \Big\|_\infty < \gamma } & \Bigg\|{ \ba{cc} {I_m - {\bf Y}_{\bf Q}^{\mathtt{md}} + (I_m - {{\bf M}^{\mathtt{md}}}){\bf Q} {\bf \Phi}_{11}^{-1} \big({\tilde {\bf N}}^{\mathtt{md}}+\Delta_{\bf \tilde{N}}) }  \\ {{\bf X}_{\bf Q}^{\mathtt{md}} + (I_m - {{\bf M}^{\mathtt{md}}}){\bf Q} {\bf \Phi}_{11}^{-1} \big({\tilde {\bf M}}^{\mathtt{md}}+\Delta_{\bf \tilde{M}}) } \ea}^{T} \Bigg\|_{\mathcal{H}_2} \\
\textrm{s.t.} \quad
    \quad &  \Bigg\|
\ba{c}  {\tilde {\bf Y}}_{\bf Q} \\  {\tilde {\bf X}}_{\bf Q} \ea \Bigg\|_{\infty} \leq \dfrac{1}{
\gamma}.
\end{aligned}
\end{equation}

\end{thm}

\begin{rem} (Validation of Constraints)
We need the initial controller in the closed loop to be robust enough to maintain the mapping between $\bf S$ and $\bf Q$ as in \thmref{QS}. Necessarily, it is considered that $\Bigg\|
\ba{c}  {\tilde {\bf Y}^{\mathtt{md}}} \\  {\tilde {\bf X}^{\mathtt{md}}} \ea \Bigg\|_{\infty}  \leq \dfrac{1}{
\gamma}$.
Furthermore, due to the fact that TFM ${\bf \Phi}_{11}({\bf Q})$ is unimodular, the standard closed-loop robust stability constraint,
 $\Bigg\|
\ba{c}  {\tilde {\bf Y}}_{\bf Q}^{\mathtt{md}} \\  {\tilde {\bf X}}_{\bf Q}^{\mathtt{md}} \ea \Bigg\|_{\infty} \leq \dfrac{1}{
\gamma}$ is also brought into consideration.
\end{rem}

\noindent
It can be seen that \eqref{RobustcostA} is actually phrased in terms of the coprime factors of the true plant, which can never be learned in practice. The standard min-max formulation from \thmref{theoremAcost} for the robust observer evaluation is non-convex by the fact that ${\bf \Phi}_{11}^{-1}$ is no longer an affine function of $\bf Q$. The duality gap renders the attempt to solve $\bf Q$ by switching the order of min and max impossible. In order to circumvent this, an upper bound on the cost functional will be derived and we formulate the robust observer evaluation problem in a Quasi-convex manner.

\begin{prop}(Quasi-Convex Formulation)\label{theoremB}
For the true plant, ${\bf G}^{\mathtt{pt}} \in \mathcal{G}_\gamma$ the robust observer evaluation problem in \eqref{RobustcostA} admits the following upper bound:
\begin{equation} \label{theoremB2}
\small
\begin{aligned}
        \min_{\alpha \in [0,1/\gamma)} & \dfrac{1}{1 - \gamma \alpha}  \min_{\bf Q\: \text{stable}}  (1-\gamma \alpha) \Big\| \ba{@{}c@{}} {\big( I_m - {\bf Y}_{\bf Q}^{\mathtt{md}}}\big) \\ {\bf X}_{\bf Q}^{\mathtt{md}} \ea ^T \Big\|_{\mathcal{H}_2} +   \Big\| \ba{@{}c@{}} I_m - {\bf M}^{\mathtt{md}} \ea  \Big\|_{\mathcal\infty}\big\| {\bf Q}  \big\|_{\mathcal{H}_2} \Big( \Big\| \ba{@{}c@{}} {\tilde {\bf N}}^{\mathtt{md}} \\ {\tilde {\bf M}}^{\mathtt{md}} \ea ^T \Big\|_{\mathcal\infty} + \gamma \Big)\\
        \centering
        &  \hspace{40pt}\textrm{s.t.} \quad
    \quad \Bigg\|
\ba{c}  {\tilde {\bf Y}_{\bf Q}^{\mathtt{md}}} \\  {\tilde {\bf X}}_{\bf Q}^{\mathtt{md}} \ea \Bigg\|_{\infty} \leq \alpha.
\end{aligned}
\end{equation} 
\end{prop}

\noindent
The inner objective function in \eqref{theoremB2} is affine in ${\bf Q}$, hence the inner optimization problem in \propref{theoremB} is convex for each fixed $\alpha$. Proof for \propref{theoremB} is provided in \hyperref[appendixB]{Appendix B}.

\begin{rem} \label{numerical} 
The Quasi-convex problem in \propref{theoremB} is formulated in frequency domain. To solve it in practice, we need to perform a Finite-Impulse Response (FIR) truncation on Markov parameters of these systems. After the FIR truncation, for each fixed $\alpha \in [0,1/\gamma)$, An equivalent Semi Definite Programming (SDP) can be formulated for the inner optimization problem, which would give us the vectorization of Markov parameters of the optimal $\bf Q^*$ to \propref{theoremB}. Details on the SDP formulation is provided in \hyperref[appendixB]{Appendix B}.
\end{rem}

\section{Analysis of End-to-End Performance}\label{endtoendanalysis}

The performance of the Robust Linear Observer from \eqref{RobustcostA}, together with the fixed state feedback $\hat{u} = F \hat{x}$ will be considered in this section. Denote the $\mathcal{H}_2$-cost of applying the control inputs $\hat{u} = F \hat{x}$ and ${u} = F {x}$ by $J_{\hat{u}}$ and $J_{{u}}$ respectively. Then it is shown in \hyperref[appendixC]{Appendix C} that: 

\begin{equation}
    J_{\hat{u}} - J_{{u}} \leq \sum \limits_{k=1}^m \{  \sum \limits_{t=0}^{\infty} [ (Fx_t - F{\hat{x}_t})^T (Fx_t - F{\hat{x}_t}) ]; \quad w_t = e_k \delta_t \},
\end{equation}
\noindent
where $e_k$ represents the $k^{\text{th}}$ standard basis vector in $\mathbb{R}^m$ and $\delta_t$ is the discrete Dirac impulse function.
Then, by the upperbound in \propref{theoremB}, we get that:

\begin{equation} \label{costJujuhat}
\small
    \begin{aligned}
    J_{\hat{u}}  -&  J_{{u}} \leq  \| \hat{u} - u \|_2^2 \\
    \leq & \Big\| \ba{cc} {\big( I_m - {\bf Y}_{\bf Q^*}^{\mathtt{md}}}\big) & {\bf X}_{\bf Q^*}^{\mathtt{md}} \ea  \Big\|_{\mathcal{H}_2} +  \Big\| \ba{@{}c@{}} I_m - {\bf M}^{\mathtt{md}} \ea  \Big\|_{\mathcal\infty} \big\| {\bf Q^*}  \big\|_{\mathcal{H}_2} \dfrac{1}{1-\gamma \alpha} \Big( \Big\| \ba{cc} {\tilde {\bf N}}^{\mathtt{md}} & {\tilde {\bf M}}^{\mathtt{md}} \ea  \Big\|_{\mathcal\infty} + \gamma \Big)
    \end{aligned}
\end{equation}
Where ${\textbf Q}^*$ is the optimizer to \propref{theoremB}.\\
Specifically, if the fixed state feedback gain $F$ happens to be the stabilizing Riccati state-feedback  $F^{\text{opt}}$, then by the virtue of separation principle the cost $J_{u}$ in \eqref{costJujuhat} becomes the optimal $\mathcal{H}_2$-cost. In this case, \eqref{costJujuhat} immediately gives a bound for the difference in $\mathcal{H}_2$-cost between the Robust Linear Controller designed from \eqref{RobustcostA}   and the optimal Linear Quadratic Regulator (LQR) for the true plant. This argumentation is deferred to \hyperref[appendixC]{Appendix C}. 

Furthermore, from \eqref{costJujuhat}, it is evident that $(J_{\hat{u}} - J_{{u}}) \sim \mathcal{O} ( \dfrac{\gamma}{1-\gamma \alpha} ) $ which indicates that the sample complexity relies heavily on $\alpha$. In practice, it is impossible to examine uncountably many $\alpha$'s in $[0,1/\gamma)$, one should pick the value of $\alpha$ empirically each time when formulating an SDP.

\begin{rem}(Feasibility)
As $\alpha < \dfrac{1}{\gamma}$ is picked manually each time to formulate a new SDP and the performance of the observer degrades much faster with a larger $\alpha$, one would like $\alpha$ to be as small as possible. However, since $\alpha$ serves as the constraint in the Quasi-convex problem, a relatively small $\alpha$ may render the feasible set empty. This implies that the robust observer performance rely on the quality of the initial controller. A better initial controller would provide not only a better fixed feedback gain, but also a larger feasible set for the inner optimization in \propref{theoremB}. 

\end{rem}

\noindent
We integrate the above results with the system identification guarantees of \cite{Zheng2022}, to provide end-to-end sample complexity bounds for learning the linear observers given a fixed feedback gain. 
Then following the system identification procedure with probability at least $(1 - \delta)$ where $\delta$ is the failure probability, it holds that
\begin{equation*}
    \| \ba{cc}  -\Delta_{\bf \tilde{M}} &  -\Delta_{\bf \tilde{N}} \ea \|_{\infty} \leq  \| \ba{cc}  {\bf {X}}^{\mathtt{md}} &  {\bf {Y}^{\mathtt{md}}} \ea \|_{\infty} 12c\beta \mathcal{R} \bigg(\sqrt{\dfrac{m\hat{d}+p\hat{d}^2+ \hat{d}log(T/\delta)}{T}}\bigg)
\end{equation*}
\noindent
Combining with the prerequisite for robustness analysis, $\| \ba{cc}  -\Delta_{\bf \tilde{M}} &  -\Delta_{\bf \tilde{N}} \ea \|_{\infty} < \gamma$ as in Assumption 1, it is reasonable to consider that the robustness radius $\gamma$ is at the level  $\mathcal{O}\Big(\sqrt{\dfrac{logT}{T}}\Big)$.

\begin{thm}\label{finaltheorem}
Define ${s}$ = $144 \Big\|\ba{cc} {\bf X}^\mathtt{md} & {\bf Y}^\mathtt{md} \ea \Big\|_{\infty}^2 c^2 \beta^2 \mathcal{R}^2$.
Then, the error in $\mathcal{H}_2$ cost of applying the control laws $\hat{u} = F \hat{x}$ and ${u} = F {x}$ is bounded as in \eqref{costJujuhat} with probability at least $(1 - \delta)$ provided that $T \geq \max\{T_s, T_*({\delta}) \}$.
Here, $T_s$ takes the larger value between $0$ and the right most zero of $\gamma^2 T - s\hat{d}log(T/\delta) - s(m\hat{d}+p\hat{d}^2)$, and $T_*(\delta)$ = $\inf \{ T | d_*(T, \delta) \in \mathcal{D}(T), d_*(T,\delta) \leq 2d_*  (\frac{T}{256}, \delta) \}$,
\noindent
where, $d_*(T,\delta)$ = $\inf \{ d | 16 \beta \mathcal{R} \alpha(d) \geq \Big\| \mathcal{\hat{H}}_{0,d,d} - \mathcal{\hat{H}}_{0,\infty,\infty}\Big\|_2 \}$,

\noindent
$\mathcal{D}(T) = \{ d \in \mathbb{N} | d \leq \dfrac{T}{cm^2 log^3(Tm/\delta)} \} $ and $f(d) = \sqrt{d}.\big(\sqrt{\dfrac{m+dp+log(T/\delta)}{T}}\big)$.
\end{thm}

\noindent
Combining \thmref{finaltheorem} with \eqref{costJujuhat}, it follows that with high probability the difference  $J_{\hat{u}}$ and  $J_{{u}}$  behaves as
\begin{equation*}
     J_{\hat{u}} - J_{{u}}  \sim \mathcal{O}\Bigg( \dfrac{\sqrt{\dfrac{logT}{T}}}{1 - \alpha \sqrt{\dfrac{logT}{T}}} \Bigg)
\end{equation*}
\noindent


\section{Conclusion and Future work}\label{conclusion}

In this paper, we have provided the sample complexity bounds for an observer-based robust LQG regulator synthesis procedure for an unknown plant,  where uncertainty is modeled as additive perturbations on the coprime factors. We combined finite-time, non-parametric LTI system identification (\cite{Sarkar2019}) with the Youla parameterization for observer performance evaluation given a fixed state feedback gain.\\
\noindent
As an opened avenue for future research is the online learning of the observer-based LQG controller under the same type of model uncertainty. One possible direction is to work out the sample complexity for online learning for: {\em (a)} the optimal state feedback (LQR) in tandem with {\em (b)} the optimal state-observer (Kalman Filter (\cite{matni2020}))  for a potentially unstable system.

\vskip 0.2in
\bibliography{main}

\begin{thebibliography}{34}
\providecommand{\natexlab}[1]{#1}
\providecommand{\url}[1]{\texttt{#1}}
\expandafter\ifx\csname urlstyle\endcsname\relax
  \providecommand{\doi}[1]{doi: #1}\else
  \providecommand{\doi}{doi: \begingroup \urlstyle{rm}\Url}\fi

\bibitem[Afri et~al.(2017)Afri, Andrieu, Bako, and Dufour]{Afri2017}
Chouaib Afri, Vincent Andrieu, Laurent Bako, and Pascal Dufour.
\newblock State and parameter estimation: A nonlinear luenberger observer
  approach.
\newblock \emph{IEEE Transactions on Automatic Control}, 62, 2017.

\bibitem[Alazard and Apkarian(1999)]{Apkarian1999}
Daniel Alazard and Pierre Apkarian.
\newblock Exact observer-based structures for arbitrary compensators.
\newblock \emph{International Journal of Robust and Nonlinear Control},
  9:\penalty0 101–118, 1999.

\bibitem[Alessandri and Coletta(2001)]{Alessandri2001}
Angelo Alessandri and Paolo Coletta.
\newblock Design of luenberger observers for a class of hybrid linear systems.
\newblock \emph{International Workshop on Hybrid Systems: Computation and
  Control}, 2001.

\bibitem[Anderson(1998)]{Anderson1998}
Brian~D.O. Anderson.
\newblock From youla–kucera to identification, adaptive and nonlinear
  control.
\newblock \emph{Automatica}, 34\penalty0 (12):\penalty0 1485--1506, 1998.

\bibitem[Bernard and Andrieu(2019)]{Bernard2019}
Pauline Bernard and Vincent Andrieu.
\newblock Luenberger observers for nonautonomous nonlinear systems.
\newblock \emph{IEEE Transactions on Automatic Control}, pages 270 -- 281,
  2019.

\bibitem[Boczar et~al.(2018)Boczar, Matni, and Recht]{Boczar2018}
Ross Boczar, Nikolai Matni, and Benjamin Recht.
\newblock Finite-data performance guarantees for the output-feedback control of
  an unknown system.
\newblock \emph{IEEE Conference on Decision and Control (CDC)}, page
  2994–2999, 2018.
\newblock \doi{https://doi.org/10.1109/CDC.2018.8618658}.

\bibitem[Dean et~al.(2018)Dean, Mania, Matni, Recht, and Tu]{Dean2018}
Sarah Dean, Horia Mania, Nikolai Matni, Benjamin Recht, and Stephen Tu.
\newblock Regret bounds for robust adaptive control of the linear quadratic
  regulator.
\newblock \emph{Conference on Neural Information Processing Systems}, page
  4188–4197, 2018.

\bibitem[Dean et~al.(2020)Dean, Mania, Matni, Recht, and Tu]{dean2020}
Sarah Dean, Horia Mania, Nikolai Matni, Benjamin Recht, and Stephen Tu.
\newblock On the sample complexity of the linear quadratic regulator.
\newblock \emph{Foundations of Computational Mathematics}, 20:\penalty0
  633–679, August 2020.
\newblock \doi{10.1007/s10208-019-09426-y}.

\bibitem[Ding et~al.(1990)Ding, Frank, and Guo]{Ding1990}
X.~Ding, P.~M. Frank, and L.~Guo.
\newblock Robust observer design via factorization approach.
\newblock \emph{IEEE Conference on Decision and Control}, 1990.

\bibitem[Ding et~al.(1994)Ding, Guo, and Frank]{Ding1994}
X.~Ding, L.~Guo, and P.~M. Frank.
\newblock Parameterization of linear observers and its application to observer
  design.
\newblock \emph{IEEE Transactions on Automatic Control}, 39\penalty0
  (8):\penalty0 1648 -- 1652, 1994.

\bibitem[Douglas(1972)]{douglas1972}
Ronald~G. Douglas.
\newblock \emph{Banach Algebra Techniques in Operator Theory}.
\newblock Springer, 1972.

\bibitem[Einicke and White(1999)]{Einicke1999}
Garry Einicke and Langford~B White.
\newblock Robust extended kalman filtering.
\newblock \emph{\url{https://arxiv.org/abs/2109.14164}}, page 2596–2599,
  1999.

\bibitem[Ghaoui and Calafiore(2001)]{Ghaoui2001}
Laurent~El Ghaoui and Giuseppe Calafiore.
\newblock Robust filtering for discrete-time systems with bounded noise and
  parametric uncertainty.
\newblock \emph{IEEE Transactions on Automatic Control}, 46:\penalty0
  1084–1089, 2001.

\bibitem[Gu and Poon(2001)]{Gu2001}
Da-Wei Gu and Fu~Wah Poon.
\newblock A robust state observer scheme.
\newblock \emph{IEEE Transactions on Automatic Control}, 46\penalty0
  (2):\penalty0 1958--1963, 2001.

\bibitem[Ionescu et~al.(1999)Ionescu, Oara, and Weiss]{robustcontrol}
Vlad Ionescu, Cristian Oara, and Martin Weiss.
\newblock \emph{Generalized Riccati Theory and Robust Control. A Popov Function
  Approach}.
\newblock Wiley, 1999.

\bibitem[Kim et~al.(2016)Kim, Shim, and Cho]{kim2016}
Taekyoo Kim, Hyungbo Shim, and Dongil~Dan Cho.
\newblock Distributed luenberger observer design.
\newblock \emph{Conference on Decision and Control (CDC)}, 2016.

\bibitem[Lee and Lamperski(2020)]{Lee2020}
Bruce Lee and Andrew Lamperski.
\newblock Non-asymptotic closed-loop system identification using autoregressive
  processes and hankel model reduction.
\newblock \emph{IEEE Conference on Decision and Control (CDC)}, 2020.

\bibitem[Levy and Nikoukhah(2012)]{Levy2012}
Bernard~C Levy and Ramine Nikoukhah.
\newblock Robust state space filtering under incremental model perturbations
  subject to a relative entropy tolerance.
\newblock \emph{IEEE Transactions on Automatic Control}, 58:\penalty0
  682–695, 2012.

\bibitem[Luenberger(1966)]{Luenberger1966}
David~G. Luenberger.
\newblock Observers for multivariable systems.
\newblock \emph{IEEE Transactions on Automatic Control}, 11\penalty0
  (2):\penalty0 190--197, 1966.

\bibitem[Mania et~al.(2019)Mania, Tu, and Recht]{mania2019}
Horia Mania, Stephen Tu, and Benjamin Recht.
\newblock Certainty equivalence is efficient for linear quadratic control.
\newblock \emph{\url{https://arxiv.org/abs/1902.07826v2}}, 2019.

\bibitem[Niazi et~al.(2022)Niazi, Cao, Sun, Das, and Johansson]{Niazi2022}
Muhammad Umar~B. Niazi, John Cao, Xudong Sun, Amritam Das, and Karl~Henrik
  Johansson.
\newblock Learning-based design of luenberger observers for autonomous
  nonlinear systems.
\newblock \emph{\url{https://arxiv.org/abs/2210.01476}}, 2022.

\bibitem[Sarkar and Rakhlin(2019)]{Sarkar2019}
Tuhin Sarkar and Alexander Rakhlin.
\newblock Near optimal finite time identification of arbitrary linear dynamical
  systems.
\newblock \emph{International Conference on Machine Learning}, 97:\penalty0
  5610--5618, 2019.

\bibitem[Sarkar et~al.(2020)Sarkar, Rakhlin, and Dahleh]{Dahleh2020}
Tuhin Sarkar, Alexander Rakhlin, and Munther~A. Dahleh.
\newblock Nonparametric finite time lti system identification.
\newblock \emph{\url{https://arxiv.org/abs/1902.01848}}, 2020.

\bibitem[Sayed(2001)]{Sayed2001}
Ali~H Sayed.
\newblock A framework for state-space estimation with uncertain models.
\newblock \emph{IEEE Transactions on Automatic Control}, 46:\penalty0
  :998–1013, 2001.

\bibitem[Tsiamis and Pappas(2019)]{papas2019}
Anastasios Tsiamis and George~J. Pappas.
\newblock Finite sample analysis of stochastic system identification.
\newblock \emph{arXiv:1903.09122v1}, 2019.

\bibitem[Tsiamis et~al.(2020)Tsiamis, Matni, and Pappas]{matni2020}
Anastasios Tsiamis, Nikolai Matni, and George~J. Pappas.
\newblock Sample complexity of kalman filtering for unknown systems.
\newblock \emph{2nd Annual Conference on Learning for Dynamics and Control},
  120, 2020.

\bibitem[Vidyasagar(1985)]{vidyasagar1985}
M.~Vidyasagar.
\newblock \emph{Control System Synthesis: A Factorization Approach}.
\newblock Cambridge, MA: MIT Press, Signal Processing, Optimization, and
  Control Series, 1985.

\bibitem[Wang and Gao(2003)]{Wang2003}
Weiwen Wang and Zhiqiang Gao.
\newblock A comparison study of advanced state observer design techniques.
\newblock \emph{American Control Conference}, 2003.

\bibitem[Wang et~al.(2015)Wang, You, and Matni]{Matni2015}
Yuh-Shyang Wang, Seungil You, and Nikolai Matni.
\newblock Localized distributed kalman filters for large-scale systems.
\newblock \emph{IFAC-PapersOnLine}, 48:\penalty0 52–57, 2015.

\bibitem[Xie and Soh(1994)]{Xie1994}
Lihua Xie and Yeng~Chai Soh.
\newblock Robust kalman filtering for uncertain systems.
\newblock \emph{Systems and Control Letters, Elsevier}, 22:\penalty0 123–129,
  1994.

\bibitem[Zhang et~al.(2021)Zhang, Ukil, Neimand, Sabau, and
  Hohil]{Zheng2022arxiv}
Yifei Zhang, Sourav Ukil, Ephraim Neimand, Serban Sabau, and Myron Hohil.
\newblock Sample complexity of the robust lqg regulator with coprime factors
  uncertainty.
\newblock \emph{\url{https://arxiv.org/abs/2109.14164}}, 2021.

\bibitem[Zhang et~al.(2022)Zhang, Ukil, Neimand, Sabau, and Hohil]{Zheng2022}
Yifei Zhang, Sourav Ukil, Ephraim Neimand, Serban Sabau, and Myron Hohil.
\newblock Sample complexity of the robust lqg regulator with coprime factors
  uncertainty.
\newblock \emph{Learning for Dynamics and Control Conference}, pages 943--953,
  2022.

\bibitem[Zheng et~al.(2020)Zheng, Furieri, Kamgarpour, and Li]{furieri2020}
Yang Zheng, Luca Furieri, Maryam Kamgarpour, and Na~Li.
\newblock Sample complexity of linear quadratic gaussian (lqg) control for
  output feedback systems.
\newblock \emph{\url{https://arxiv.org/abs/2011.09929}}, pages 1--33, 2020.

\bibitem[Zhou et~al.(1996)Zhou, Doyle, and Glover]{optimalcontrol}
Kemin Zhou, John~Comstock Doyle, and Keith Glover.
\newblock \emph{Robust and optimal control}, volume~40.
\newblock Prentice hall, New Jersey, 1996.

\end{thebibliography}




\tableofcontents
\addcontentsline{toc}{section}{Appendix}
\renewcommand{\theHsection}{A\arabic{section}}

\section*{Appendix}\label{appendix}
\noindent
This appendix is divided into following parts. 
 \hyperref[appendixA]{Appendix A} presents a brief review of related works on observer parameterization. Also, a handful of mathematical preliminaries on norm identities and inequalities (\cite{optimalcontrol}) are provided here. 
 An overview of the closed loop mapping is stated in \hyperref[appendixB]{Appendix B} along with robust synthesis proofs for \lemref{phi11phi22}, \thmref{oarasthm1} and \thmref{theoremAcost}. 
 This appendix section also completes the suboptimality guarantee proof in \propref{theoremB}. 
 \hyperref[appendixC]{Appendix C} presents a brief overview on $\mathcal{H}_2$ optimal cost. 
Non-aymptotic closed loop system identification (\cite{Dahleh2020})  is discussed in \hyperref[appendixD]{Appendix D}.



\appendix

\section{Related Works}\label{appendixA}

Recent years have seen a significant amount of research work focused on the finite time (non-asymptotic) learning of the optimal LQ regulator for a "unknown" plant utilizing the modern optimization methods and statistical tools from the learning framework. For related work on the \textit{Identification of Dynamical Systems}, \textit{Controller Design}, \textit{Robust Control} and \textit{Optimal Control} we refer to the Appendix A of \cite{Zheng2022arxiv}.\\

\textit{Observer Design}: Using the factorization technique for the parameterization of linear observers and associated estimation error dynamics is a classical result. This outcome offers a dual representation of the popular linear controller parameterization and offers fresh information on observer design (\cite{Ding1990}) that may be applied to both robust observer design and observer construction. \cite{Ding1994} outline and address issues with design and parameterization of robust linear observers in the frequency domain.
Methods for determining any compensator's observer-based or LQG form with arbitrary order are explored in \cite{Apkarian1999}.
The performance and attributes of advanced state observers are compared in \cite{Wang2003}. These observers were first put out as a solution to the traditional observers' reliance on a precise mathematical representation of the plant, such as the Kalman filter and the Luenberger observer.
Classical approaches to robust Kalman Filtering can be found in \cite{Xie1994}, \cite{Einicke1999}, \cite{Ghaoui2001}, \cite{Sayed2001}, \cite{Levy2012}, where parametric uncertainty is explicitly taken into account during the kalman filter synthesis procedure. The process of designing a Kalman filter for an unknown or partially observed autonomous linear time-invariant system has been discussed in \cite{matni2020}, which was the first end-to-end sample complexity bounds for an unidentified system's Kalman filtering.
The strategy of Luenberger observer is investigated to suggest a solution to the state and parameter estimation for dynamical systems in \cite{Alessandri2001}, \cite{kim2016}, \cite{Afri2017}, \cite{Bernard2019}, \cite{Niazi2022}. When the plant dynamics are relatively well understood, the Luenberger observer performs well, but the estimation of the states might not be precise enough in the presence of model perturbations. By utilizing the Lyapunov stability theorem, \cite{Gu2001} derived a novel resilient observer technique to overcome this problem. More on robust control has been discussed in Appendix A of \cite{Zheng2022arxiv}.\\

\textit{Norm/Inequality preliminaries}: Useful $\mathcal{H}_2$ and $\mathcal{H}_{\infty}$ norm identities and inequalities have been adapted from \cite{optimalcontrol} which are essential for the proofs. For more details on this we refer to Appendix B of \cite{Zheng2022arxiv}.

\section{Details on Proofs}\label{appendixB}

\subsection{Closed Loop Details}

\begin{proof}  \textbf{for \thmref{dingtheorem}}:
Without loss of generality, by ommiting the additive disturbance $\delta_k$ in the state equation \eqref{stateeq} we get:
\begin{equation} \label{stateeqwoadtvdisturbance}
    x(z) = (zI-A)^{-1} B \big(u(z) + w(z) \big)
\end{equation}
Lets define a pseudo-state vector $\epsilon (z)$ by $\textbf{M} (z) \epsilon (z) = u(z)$. Then system \eqref{stateeq} can be expressed as:
\begin{equation} \label{pseudostateeqn}
    \begin{aligned}
        u(z) & = \textbf{M}(z) \epsilon(z) \\
        y(z) & = \textbf{N}(z) \epsilon(z) + \textbf{N}(z) \textbf{M}(z)^{-1} w(z) + \nu (z)
    \end{aligned}
\end{equation}
\noindent
Denote $\textbf{P}(z) = (zI-A_F)^{-1}B$. Then by applying the rules of connecting stable systems and $\textbf{M}(z) = F(zI-A_F)^{-1}B+I$, it's possible to deduce that $\textbf{P}(z)\textbf{M}(z)^{-1} = F(zI-A_F)^{-1}B+I$.

\noindent
By plugging back \eqref{pseudostateeqn} into \eqref{stateeqwoadtvdisturbance} and \eqref{ses1}, we get the following:
\begin{equation} \label{observerprrofeqn}
    \begin{aligned}
        x(z) = \textbf{P}(z) \epsilon (z) &+ \textbf{P}(z) \textbf{M}(z)^{-1} w(z) \\
        \hat{x}(z) =  {{\bf \Psi}^{u}}(z) \textbf{M}(z) \epsilon(z)   + {{\bf \Psi}^{y}} (z) \textbf{N}(z) \epsilon(z) &+ 
         {{\bf \Psi}^{y}}(z) \textbf{N}(z)  \textbf{M}(z)^{-1} w(z)   + {{\bf \Psi}^{y}}\nu (z)
    \end{aligned}
\end{equation}
\noindent
From \eqref{observerprrofeqn}, it's clear that $\hat{x}(z) = {{\bf \Psi}^{u}}(z) u(z) + {{\bf \Psi}^{y}}(z) y(z)$ is an observer of \eqref{stateeq} if and only if :
(i) the transfer function from $\epsilon (z)$ to $x(z)$ matches with the transfer function from $\epsilon (z)$ to $\hat{x}(z)$ and (ii) the transfer functions from noises $w(z)$ and $\nu (z)$ to the estimation error $(x(z) - \hat{x}(z))$ are stable.

\noindent
Clearly, (i) gives us ${{\bf \Psi}^{u}}(z) {\textbf M}(z) + {{\bf \Psi}^{y}}(z) {\textbf N}(z) = {\textbf P}(z)$. (ii) is always satisfied as ${{\bf \Psi}^{u}}(z)$ and ${{\bf \Psi}^{y}}(z)$ follow the parameterization in 
\eqref{Psi1Psi2}. 

\noindent
Combining (i) and (ii), after inverse z-transform, it's evident that:
\begin{equation} \label{estimationerrorlimitto0}
    \lim_{k\to\infty}  (x_k - \hat{x}_k) = 0,
\end{equation}
which is the definition of an observer and this completes the proof.
\end{proof}

\begin{proof} \textbf{for \thmref{TransferfromNoisetoError}}:
The transfer function matrices ${{\bf \Psi}^{u}_{\bf S}}(z)$ and ${{\bf \Psi}^{y}_{\bf S}}(z)$ in \eqref{Psi1Psi2} are stable as $\textbf{S}(z) \in \mathbb{R}(z)^{n \times p}$ is stable. It's possible to use \eqref{Psi1Psi2} instead of \eqref{observerparaneterizationwoS} since the parameterization in \eqref{Psi1Psi2} would always satisfy \eqref{observerparaneterizationwoS} and reveal a valid observer for each stable $\textbf{S}(z) \in \mathbb{R}(z)^{n \times p}$  if the B\'{e}zout identity\eqref{bezoutidentity} holds.

Now, we consider the transfer function from the noises to the estimation error. From \eqref{observerprrofeqn}, we can conclude that
\begin{equation*}
\begin{aligned}
    T^{(x-\hat{x})w} & = {\textbf{P}}(z)  \textbf{M}(z)^{-1}  - {{\bf \Psi}^{y}}(z) {\textbf N}(z){\textbf M}(z)^{-1} = {\textbf{P}}(z)  \textbf{M}(z)^{-1} - {{\bf \Psi}^{y}}(z) {\tilde{\textbf M}}(z)^{-1} {\tilde{\textbf N}}(z)  \\
     & = {\textbf{P}}(z)  \textbf{M}(z)^{-1} - {\textbf{P}}(z)  \textbf{X}(z)  {\tilde{\textbf M}}(z)^{-1} {\tilde{\textbf N}}(z) + {\textbf{S}}(z) {\tilde{\textbf M}}(z) {\tilde{\textbf M}}(z)^{-1} {\tilde{\textbf N}}(z)\\
     & = {\textbf{P}}(z)  \textbf{M}(z)^{-1} \big(I -  \textbf{M}(z) \textbf{Y}(z)\big) + {\textbf{P}}(z)  \textbf{Y}(z) - {\textbf{P}}(z)  \textbf{X}(z)  {\tilde{\textbf M}}(z)^{-1} {\tilde{\textbf N}}(z) + {\textbf{S}}(z) {\tilde{\textbf N}}(z)\\
     & = {\textbf{P}}(z)  \textbf{M}(z)^{-1} {\tilde{\textbf X}}(z) {\tilde{\textbf N}}(z) - {\textbf{P}}(z)  \textbf{X}(z)  {\tilde{\textbf M}}(z)^{-1} {\tilde{\textbf N}}(z) + {\textbf{P}}(z)  \textbf{Y}(z) + {\textbf{S}}(z) {\tilde{\textbf N}}(z)\\
     & = {{\bf \Psi}^{u}_{\bf S}}(z)\\
    T^{(x-\hat{x})\nu} & =  -{{\bf \Psi}^{y}_{\bf S}}(z).
\end{aligned}
\end{equation*}

\end{proof}

\begin{proof} \textbf{for \thmref{QS}}: Since $\textbf{Q} \in \mathbb{R}(z)^{m \times p}$ and $\textbf{S} \in \mathbb{R}(z)^{n \times p}$ are stable TFMs, without loss of generality it can be assumed that, for any \textit{onto} linear mapping ${K} \in \mathbb{R}^{m \times n}$, there exists $\textbf{E} (z) \in  \mathbb{R}(z)^{m \times p}$, such that 
\begin{equation} \label{linearontomapping1}
    \textbf{Q}(z) = K \textbf{S}(z) + \textbf{E} (z)
\end{equation}
\noindent
To go from $\textbf{S} (z)$ to $\textbf{Q} (z)$, we need to solve $\textbf{E} (z)$ for a pre-specified ${K} \in \mathbb{R}^{m \times n}$. In order to get a meaniful result, here we choose that $K=F$. Then from \eqref{YoulaEq} and \eqref{STES}, we get the following:
\begin{equation} \label{linearontomapping2}
    \begin{aligned}
        ( {\textbf Y} - {\textbf Q} {\tilde{\textbf N}} ) u  = & (-{\textbf X} - {\textbf Q} {\tilde{\textbf M}}) y \\
        u = F \hat{x}  = F ({\textbf{PY}}+ & \textbf{S} {\tilde{\textbf N}})u + F( {\textbf{PX}}-{\textbf{S}} {\tilde{\textbf{M}}} ) y
    \end{aligned}
\end{equation}
\noindent
The second inequality in \eqref{linearontomapping2} provides us that:
    \begin{equation*}
        (I - F(\textbf{PY}+\textbf{S}\tilde{\textbf N})) u = F (\textbf{PX} - \textbf{S}\tilde{\textbf M})y
    \end{equation*}
\noindent
It follows that $-(I - F(\textbf{PY}+\textbf{S}\tilde{\textbf N})){\textbf X}_{\textbf Q} = F (\textbf{PX} - \textbf{S}\tilde{\textbf M}){\textbf Y}_{\textbf Q}$. Solving this yields that for $K=F$ in \eqref{linearontomapping1}, the corresponding ${\textbf{E}}(z)$ is just $- {\tilde{\textbf{X}}}(z)$.Thus if $F \in \mathbb{R}^{m \times n}$ is onto, it's always possible to write ${\text F} \textbf{S}(z) = \textbf{Q}(z) + \tilde{\textbf{X}}(z)$ which completes the proof.

\end{proof}

\begin{proof} \textbf{for \propref{optimalcontrolproblem}}:
The Optimal Observer Evaluation Problem due to \remref{Fdesign} below is:
\begin{equation*}
    \min_{\bf Q\: \text{stable}}  \Bigg\| \ba{cc} { FT^{(x-\hat{x})w}} & { FT^{(x-\hat{x})\nu}} \ea \Bigg\|_{\mathcal{H}_2}
\end{equation*}
\noindent
Now, the transformed error for the observer evaluation problem above follow from the relationships ${\text F}{\bf S}(z) = {\bf Q}(z) + \tilde{\bf X}(z) = {\bf Q}(z) + \tilde{\bf X}_{\bf Q}(z) - {\bf M}(z){\bf Q}(z)$ and ${\text F}{\bf P}(z) = {\bf M}(z) - { I_m}$  as below:
\begin{equation*}
\begin{aligned}
    FT^{(x-\hat{x})w} & = {\text F}{{\bf \Psi}_{\bf Q}^{u}}(z) = {\text F} \big({\bf P}(z) {\bf Y}_{\bf Q}(z) + {\bf S}(z) \tilde{\bf N}(z) \big)\\
    & = {\text F} {\bf P}(z) {\bf Y}_{\bf Q}(z) + {\text F} {\bf S}(z) \tilde{\bf N}(z) \\
    & =  \big({\bf M}(z) - { I_m}\big){\bf Y}_{\bf Q}(z) + \big({\bf Q}(z) + \tilde{\bf X}_{\bf Q}(z) - {\bf M}(z){\bf Q}(z) \big) \tilde{\bf N}(z) \\
    & =  \big({\bf M}(z){\bf Y}_{\bf Q}(z) + \tilde{\bf X}_{\bf Q}(z) \tilde{\bf N}(z) \big) - {\bf Y}_{\bf Q}(z) + \big({I_m} - {\bf M}(z)\big) {\bf Q}(z) \tilde{\bf N}(z) \\
    & = {I_m - {\bf Y}_{\bf Q}(z) + \big(I_m - {\bf M}(z)\big) {\bf Q}(z) \tilde{\bf N}(z)}\\
   { F}T^{(x-\hat{x})\nu} & = - {\text F}{{\bf \Psi}_{\bf Q}^{y}}(z)  = - {\text F} \big( {\bf P}(z) {\bf X}_{\bf Q}(z) - {\bf S}(z) \tilde{\bf M}(z)  \big) \\
   & = - {\text F}  {\bf P}(z) {\bf X}_{\bf Q}(z) + {\text F} {\bf S}(z) \tilde{\bf M}(z) \\ & = - \big({\bf M}(z) - { I_m}\big) {\bf X}_{\bf Q}(z) +  \big({\bf Q}(z) + \tilde{\bf X}_{\bf Q}(z) - {\bf M}(z){\bf Q}(z) \big)  \tilde{\bf M}(z)  \\
   & = \big({ - \bf M}(z){\bf X}_{\bf Q}(z) +  \tilde{\bf X}_{\bf Q}(z) \tilde{\bf M}(z)) + {\bf X}_{\bf Q}(z) + \big(I - {\bf M}(z) \big){\bf Q}(z) \tilde{\bf M}(z)    \\
   & = {{\bf X}_{\bf Q}(z) + \big(I - {\bf M}(z) \big) {\bf Q}(z) \tilde{\bf M}(z)}.
\end{aligned}
\end{equation*}
\noindent
Hence, the Optimal Observer Evaluation Problem have the form:
\begin{equation*}
 \min_{\bf Q\: \text{stable}}  \Bigg\| \ba{cc} { FT^{(x-\hat{x})w}} & { FT^{(x-\hat{x})\nu}} \ea \Bigg\|_{\mathcal{H}_2} = 
\end{equation*}
\begin{equation*}
 \min_{\bf Q\: \text{stable}}  \Bigg\| \ba{@{}cc@{}} {I_m - {\bf Y}_{\bf Q}(z) + \big(I_m - {\bf M}(z)\big) {\bf Q}(z) \tilde{\bf N}(z)} & {{\bf X}_{\bf Q}(z) + \big(I - {\bf M}(z) \big) {\bf Q}(z) \tilde{\bf M}(z)} \ea \Bigg\|_{\mathcal{H}_2}.
\end{equation*}
\end{proof}

\subsection{Robust Synthesis Details}
The Bezout identity is retrieved before finding the closed loop maps associated with it by using \eqref{PhiMatrix} as below:
\begin{equation*}\label{bezout4phi}
 \ba{cc}   { \bf \Phi}_{11}^{-1} &  0 \\ 0 & I_m \ea
\ba{cc}  { \bf \Phi}_{11} & 0 \\   0 & { \bf \Phi}_{22} \ea 
\ba{cc}  I_p & 0 \\  0 &  { \bf \Phi}_{22}^{-1}\ea = \ba{cc}   { I_p} &   {0} \\  {0} &  { I_m} \ea ,
\end{equation*} 

 \begin{equation*}\label{Phitoclosedloopmap}
 \ba{cc}   {\bf \Phi}_{11}^{-1} &  0 \\ 0 & I_m \ea
\ba{cc}    ({\tilde {\bf M}}^{\mathtt{md}}+\Delta_{\bf \tilde{M}}) &   ({\tilde {\bf N}}^{\mathtt{md}}+\Delta_{\bf \tilde{N}}) \\ - {{\bf X}_{\bf Q}^{\mathtt{md}}}  &  {{\bf Y}_{\bf Q}^{\mathtt{md}}}  \ea
\ba{cc}  {\tilde {\bf Y}_{\bf Q}^{\mathtt{md}}} & -({{\bf N}}^{\mathtt{md}}+\Delta_{\bf {N}}) \\   {\tilde {\bf X}_{\bf Q}^{\mathtt{md}}} &  ({{\bf M}}^{\mathtt{md}}+\Delta_{\bf {M}}) \ea
\ba{cc}  I_p & 0 \\  0 & {{\bf \Phi}}_{22}^{-1}\ea = \ba{cc}   { I_p} &   {0} \\  {0} &  { I_m} \ea ,
\end{equation*} 

 \begin{equation}\label{closedloopmap}
\small \ba{cc}    {\bf \Phi}_{11}^{-1} ({\tilde {\bf M}}^{\mathtt{md}}+\Delta_{\bf \tilde{M}}) &   {\bf \Phi}_{11}^{-1} ({\tilde {\bf N}}^{\mathtt{md}}+\Delta_{\bf \tilde{N}}) \\ - {{\bf X}_{\bf Q}^{\mathtt{md}}}  &  {{\bf Y}_{\bf Q}^{\mathtt{md}}}  \ea
\ba{cc}  {\tilde {\bf Y}_{\bf Q}^{\mathtt{md}}} & -({{\bf N}}^{\mathtt{md}}+\Delta_{\bf {N}}) {\bf \Phi}_{22}^{-1} \\   {\tilde {\bf X}_{\bf Q}^{\mathtt{md}}} &  ({{\bf M}}^{\mathtt{md}}+\Delta_{\bf {M}}) {\bf \Phi}_{22}^{-1} \ea
= \ba{cc}   { I_p} &   {0} \\  {0} &  { I_m} \ea.
\end{equation} \\

\begin{proof}\textbf{for \lemref{phi11phi22}}:
From DCF matrix ${\bf \Phi}$ in \eqref{PhiMatrix}, ${\bf \Phi}_{11}$ = $({\tilde {\bf M}^{\mathtt{md}}}+\Delta_{\bf {\tilde{M}}}) {\tilde {\bf Y}_{\bf Q}^{\mathtt{md}}}  + ({\tilde {\bf N}^{\mathtt{md}}}+\Delta_{\bf {\tilde{N}}}) {\tilde {\bf X}_{\bf Q}^{\mathtt{md}}}$ and ${\bf{\Phi}}_{22}$ = ${\bf X} ({\bf N}^{\mathtt{md}}+\Delta_{\bf N})+{\bf Y}({\bf M}^{\mathtt{md}}+\Delta_{\bf M}) $.
Next using Bezout identity for nominal model in \eqref{bezout2} it follows that
\begin{equation*}\label{phimtrixproof1}
\begin{aligned}
{\bf \Phi}_{11} &= \ba{cc} ({\tilde {\bf M}^{\mathtt{md}}}+\Delta_{\bf {\tilde{M}}}) &  ({\tilde {\bf N}^{\mathtt{md}}}+\Delta_{\bf {\tilde{N}}})\ea
\ba{c} {\tilde {\bf Y}_{\bf Q}^{\mathtt{md}}} \\  {\tilde {\bf X}_{\bf Q}^{\mathtt{md}}} \ea\\
&= \ba{cc}  {\tilde {\bf M}^{\mathtt{md}}} &  {\tilde {\bf N}^{\mathtt{md}}} \ea
\ba{c} {\tilde {\bf Y}_{\bf Q}^{\mathtt{md}}} \\ {\tilde {\bf X}_{\bf Q}^{\mathtt{md}}} \ea + \ba{cc}  \Delta_{\bf \tilde{M}} &  \Delta_{\bf \tilde{N}} \ea
\ba{c}  {\tilde {\bf Y}_{\bf Q}^{\mathtt{md}}} \\  {\tilde {\bf X}_{\bf Q}^{\mathtt{md}}} \ea = I_p + \ba{cc}  \Delta_{\bf \tilde{M}} &  \Delta_{\bf \tilde{N}} \ea
\ba{c}  {\tilde {\bf Y}_{\bf Q}^{\mathtt{md}}} \\  {\tilde {\bf X}_{\bf Q}^{\mathtt{md}}} \ea ;
\end{aligned}
\end{equation*}

\begin{equation*}\label{phimtrixproof2}
\begin{aligned}
{\bf \Phi}_{22} &= \ba{cc} {\bf X_{\bf Q}^{\mathtt{md}}} &  {\bf Y_{\bf Q}^{\mathtt{md}}} \ea
\ba{c} ({\bf N}^{\mathtt{md}}+\Delta_{\bf N}) \\  ({\bf M}^{\mathtt{md}}+\Delta_{\bf M}) \ea\\
&= \ba{cc}  {\bf X_{\bf Q}^{\mathtt{md}}} &  {\bf Y_{\bf Q}^{\mathtt{md}}} \ea
\ba{c} {\bf N}^{\mathtt{md}} \\ {\bf M}^{\mathtt{md}} \ea + \ba{cc}  {\bf X_{\bf Q}^{\mathtt{md}}} &  {\bf Y_{\bf Q}^{\mathtt{md}}} \ea
\ba{c}  \Delta_{\bf N} \\ \Delta_{\bf M} \ea = I_m +  \ba{cc}  {\bf X_{\bf Q}^{\mathtt{md}}} &  {\bf Y_{\bf Q}^{\mathtt{md}}} \ea \ba{c}  \Delta_{\bf N} \\ \Delta_{\bf M} \ea .
\end{aligned}
\end{equation*}
\end{proof}
\noindent
Proof of \thmref{oarasthm1} directly follows from \cite{Zheng2022arxiv}.
Before giving the proof for \thmref{oarasthm1}, the small gain theorem is stated here.
\noindent
\begin{thm}[Small Gain Theorem]\label{smallgaintheorem} \cite[Theorem~7.4.1/ page 225 ]{robustcontrol}
Let ${\bf G}_1 \in \mathbb{R}(z)^{p\times m}$ and ${\bf G}_2 \in \mathbb{R}(z)^{m \times p}$ be two TFM's respectively. If $\|{\bf G}_1\|_{\infty} \leq \dfrac{1}{ \gamma}$ and $\|{\bf G}_2\|_{\infty} \leq \gamma$, for some $\gamma > 0$, then the closed loop feedback system of ${\bf G}_1$ and ${\bf G}_2$ is internally stable.
\end{thm}

\begin{proof}\textbf{for \thmref{oarasthm1}}:
For any stable ${\bf Q}$ satisfying $\small \Bigg\| \ba{c}  {\tilde {\bf Y}}_{\bf Q} \\  {\tilde {\bf X}}_{\bf Q} \ea \Bigg\|_{\infty}  \leq \dfrac{1}{ \gamma}$ it follows that ${\bf \Phi}_{11} = \bigg({I_p} +$ $\small \ba{cc}  \Delta_{\bf \tilde{M}} &  \Delta_{\bf \tilde{N}} \ea \ba{c}  {\tilde {\bf Y}}_{\bf Q} \\  {\tilde {\bf X}}_{\bf Q} \ea \bigg)$ is unimodular (square and stable with a stable inverse) due to the fact that: {\em (a)} The term  $\bigg(I_p + \small \ba{cc}  \Delta_{\bf \tilde{M}} &  \Delta_{\bf \tilde{N}} \ea \ba{c}  {\tilde {\bf Y}}_{\bf Q} \\  {\tilde {\bf X}}_{\bf Q} \ea\bigg)$ is stable since all factors are stable and {\em (b)} We know that $\small \Big\| \ba{cc}  \Delta_{\bf \tilde{M}} &  \Delta_{\bf \tilde{N}} \ea \Big\|_{\infty}  < \gamma$ from the definition of the Model Uncertainty Set. At the same time $\Bigg( I_p + \small \ba{cc}  \Delta_{\bf \tilde{M}} &  \Delta_{\bf \tilde{N}} \ea \ba{c}  {\tilde {\bf Y}}_{\bf Q} \\  {\tilde {\bf X}}_{\bf Q} \ea \Bigg)^{-1}$ is guaranteed to be stable via the Small Gain Theorem.\\


\noindent
Conversely, if a Youla parameter ${\bf Q}$  yields a $\gamma$-robustly stabilizable controller of the nominal model  then necessarily $\small \Bigg\| \ba{c}  {\tilde {\bf Y}}_{\bf Q} \\  {\tilde {\bf X}}_{\bf Q} \ea \Bigg\|_{\infty} \leq \dfrac{1}{ \gamma}$.
The proof of this claim is done by contradiction. Assume that $\small \Bigg\| \ba{c}  {\tilde {\bf Y}}_{\bf Q} \\  {\tilde {\bf X}}_{\bf Q} \ea \Bigg\|_{\infty}  > \dfrac{1}{\gamma}$.
Then by the Spectral Mapping Theorem \cite[page~41-42]{douglas1972} there must exist $\small \Big\| \ba{cc}  \Delta_{\bf \tilde{M}} &  \Delta_{\bf \tilde{N}} \ea \Big\|_{\infty}  < \gamma$ such that  ${\bf \Phi}_{11} = \bigg( {I_p} + $ $\small \ba{cc}  \Delta_{\bf \tilde{M}} &  \Delta_{\bf \tilde{N}} \ea \ba{c}  {\tilde {\bf Y}}_{\bf Q} \\  {\tilde {\bf X}}_{\bf Q} \ea \bigg)$ is not unimodular and consequently the Youla parameter ${\bf Q}$ does not produce an $\gamma$-robustly stabilizable controller, which is a contradiction. The proof ends. 
\end{proof}

\begin{proof} \textbf{for \thmref{theoremAcost}}:
Let's denote ${{\bf \Psi}_{\mathtt{pt}}^{u}}$ and $-{{\bf \Psi}_{\mathtt{pt}}^{y}}$ as the transfer from noises to the estimation error when connecting the controller ${{\bf K}_{\bf Q}^{\mathtt{md}}} = ({\bf Y}_{\bf Q}^{\mathtt{md}})^{-1} {\bf X}_{\bf Q}^{\mathtt{md}} = {\bf {\tilde{X}^{\mathtt{md}}}}_{\bf Q} ({\bf \tilde{Y}^{\mathtt{md}}}_{\bf Q})^{-1}$ to the true plant ${\bf G}^{\mathtt{pt}} = (\tilde{\bf M}^{\mathtt{pt}})^{-1} \tilde{\bf N}^{\mathtt{pt}} = {\bf N}^{\mathtt{pt}} ({\bf M}^{\mathtt{pt}})^{-1}$,
\begin{equation*}
    {{\bf \Psi}^{u}_{\mathtt{pt}}} {=} {\bf P}^{\mathtt{pt}} {\bf Y}_{\bf Q}^{\mathtt{md}} + {\bf S} \tilde{\bf{N}}^{\mathtt{pt}}, \quad
         {{\bf \Psi}^{y}_{\mathtt{pt}}} {=} {\bf P}^{\mathtt{pt}} {\bf X}_{\bf Q}^{\mathtt{md}} - {\bf S} \tilde{\bf{M}}^{\mathtt{pt}} 
\end{equation*}
\noindent
We use the norm of transformed error $\big\| \ba{cc}  F {\bf \Psi}^{u}_{\mathtt {pt}} & -F{\bf \Psi}^{y}_{\mathtt {pt}} \ea \big\|_{\mathcal{H}_2}$  to reveal the observer performance in terms of the Youla parameter $\textbf{Q}$ as below:
\begin{equation*}
\begin{aligned}
     F{{\bf \Psi}_{\mathtt{pt}}^{u}} & =  \big({\bf M}^{\mathtt{pt}} - { I_m}\big){\bf Y}_{\bf Q}^{\mathtt{md}} + \big({\bf Q} + \tilde{\bf X}^{\mathtt{md}} \big) \tilde{\bf N}^{\mathtt{pt}}\\
    & =  \big({\bf M}^{\mathtt{pt}} - { I_m}\big){\bf Y}_{\bf Q}^{\mathtt{md}} + \big({\bf Q} + \tilde{\bf X}_{\bf Q}^{\mathtt{md}} - {\bf M}^{\mathtt{md}}{\bf Q} \big) \tilde{\bf N}^{\mathtt{pt}} \\
    & =  \big({\bf M}^{\mathtt{pt}}{\bf Y}_{\bf Q}^{\mathtt{md}} + \tilde{\bf X}_{\bf Q}^{\mathtt{md}} \tilde{\bf N}^{\mathtt{pt}} \big) - {\bf Y}_{\bf Q}^{\mathtt{md}} + \big({I_m} - {\bf M}^{\mathtt{md}} \big) {\bf Q} \tilde{\bf N}^{\mathtt{pt}} \\
    & = {I_m - {\bf Y}_{\bf Q}^{\mathtt{md}} + \big(I_m - {\bf M}^{\mathtt{md}}\big) {\bf Q} \tilde{\bf N}^{\mathtt{pt}}}\\
    & = {I_m - {\bf Y}_{\bf Q}^{\mathtt{md}} + \big(I_m - {\bf M}^{\mathtt{md}}\big) {\bf Q} {\bf \Phi}_{11}^{-1} (\tilde{\bf N}^{\mathtt{md}} + \Delta_{\bf \tilde{N}})}\\
    \\
  - F{{\bf \Psi}_{\mathtt{pt}}^{y}}  & = - \big({\bf M}^{\mathtt{pt}} - { I_m}\big) {\bf X}_{\bf Q}^{\mathtt{md}} +  \big({\bf Q} + \tilde{\bf X}^{\mathtt{md}}\big)  \tilde{\bf M}^{\mathtt{pt}}  \\
   & = - \big({\bf M}^{\mathtt{pt}} - { I_m}\big) {\bf X}_{\bf Q}^{\mathtt{md}} +  \big({\bf Q} + \tilde{\bf X}_{\bf Q}^{\mathtt{md}} - {\bf M}^{\mathtt{md}}{\bf Q} \big)  \tilde{\bf M}^{\mathtt{pt}}  \\
   & = \big({ - \bf M}^{\mathtt{pt}} {\bf X}_{\bf Q}^{\mathtt{md}} +  \tilde{\bf X}_{\bf Q}^{\mathtt{md}} \tilde{\bf M}^{\mathtt{pt}}) + {\bf X}_{\bf Q}^{\mathtt{md}} + \big(I - {\bf M}^{\mathtt{md}} \big){\bf Q} \tilde{\bf M}^{\mathtt{pt}}    \\
   & = {{\bf X}_{\bf Q}^{\mathtt{md}} + \big(I - {\bf M}^{\mathtt{md}} \big) {\bf Q} {\bf \Phi}_{11}^{-1} (\tilde{\bf M}^{\mathtt{md}} + \Delta_{\bf \tilde{M}})}.
\end{aligned}
\end{equation*}
\noindent
The norm of these two transfers above are the objective function that we seek to minimize. The condition $\Bigg\| \ba{c}  {\tilde {\bf Y}}_{\bf Q}^{\mathtt{md}} \\  {\tilde {\bf X}}_{\bf Q}^{\mathtt{md}} \ea \Bigg\|_{\infty}  \leq \dfrac{1}{ \gamma}$ guarantees the robustness by providing the invertibility of ${\bf \Phi}_{11}$, it would be the constraint in the robust observer evaluation problem.

\end{proof}

\noindent
To proof \thmref{theoremB} an important lemma is stated below which is a standard result optimization.
\begin{lem} \label{sarahopt1} (\cite{dean2020})
For functions $f: X \mapsto R$ and $g: X \mapsto R$ and constraint set $C \subseteq X$, consider
\begin{equation}
    \min_{x \in C}  \dfrac{f(x)}{1 - g(x)}
\end{equation}

\noindent
Assuming that $f (x) \geq 0$ and $0 \leq g(x) < 1$, $(\forall) x \in C$, this optimization problem
can be reformulated as an outer single-variable problem and an inner-constrained optimization problem (the objective value of an optimization over the empty set is defined to be infinity):
\begin{equation}
     \min_{x \in C}  \dfrac{f(x)}{1 - g(x)} =  \min_{\delta \in [0,1)} \dfrac{1}{1 - \delta} \min_{x \in C}  \{f(x)| g(x) \leq \delta \}.
\end{equation}
\noindent
The equivalence here is established by considering dividing the feasible set $\mathcal{C}$ into two parts: $\{ x \in \mathcal{C} | g(x) \leq \delta \}$ and $\{ x \in \mathcal{C} | \delta < g(x) <1  \}$.

\end{lem}

\begin{proof} \textbf{for \propref{theoremB}}:
The objective function in \thmref{theoremAcost} admits the following upper bound:

\begin{equation*}
\begin{aligned}
\Bigg\| & \ba{cc} {I_m - {\bf Y}_{\bf Q}^{\mathtt{md}} + (I_m - {{\bf M}^{\mathtt{md}}}){\bf Q} {\bf \Phi}_{11}^{-1} \big({\tilde {\bf N}}^{\mathtt{md}}+\Delta_{\bf \tilde{N}}) }  \\ {{\bf X}_{\bf Q}^{\mathtt{md}} + (I_m - {{\bf M}^{\mathtt{md}}}){\bf Q} {\bf \Phi}_{11}^{-1} \big({\tilde {\bf M}}^{\mathtt{md}}+\Delta_{\bf \tilde{M}}) } \ea ^{T} \Bigg\|_{\mathcal{H}_2} \\
 \leq   \Bigg\|    \ba{cc} I_m - {\bf Y}_{\bf Q}^{\mathtt{md}} \\ {\bf X}_{\bf Q}^{\mathtt{md}} \ea  ^{T}  & \Bigg\|_{\mathcal{H}_2} + \Big\| I_m - {{\bf M}^{\mathtt{md}}} \Big\|_\infty  \Big\| {\bf Q}   \Big\|_{\mathcal{H}_2}  \Bigg\| {\bf \Phi}_{11}^{-1}  \ba{cc}  {\tilde {\bf N}}^{\mathtt{md}}+\Delta_{\bf \tilde{N}} \\  {\tilde {\bf M}}^{\mathtt{md}}+\Delta_{\bf \tilde{M}}   \ea ^{T}    \Bigg\|_\infty
\end{aligned}
\end{equation*}

The last part of the upper bound above is formulated as follows:
\begin{equation*}
\begin{aligned}
\Bigg\| & {\bf \Phi}_{11}^{-1}  \ba{cc}  {\tilde {\bf N}}^{\mathtt{md}}+\Delta_{\bf \tilde{N}} \\  {\tilde {\bf M}}^{\mathtt{md}}+\Delta_{\bf \tilde{M}}   \ea   ^{T}  \Bigg\|_\infty \leq  \Big\| {\bf \Phi}_{11}^{-1} \Big\|_\infty  \Bigg( \Bigg\|  \ba{cc}  {\tilde {\bf N}}^{\mathtt{md}} \\  {\tilde {\bf M}}^{\mathtt{md}}  \ea  ^{T}   \Bigg\|_\infty + \gamma \Bigg) \\
= & \Bigg\| {\Bigg({ I_p} - \ba{cc}  -\Delta_{\bf \tilde{M}} &  -\Delta_{\bf \tilde{N}} \ea \ba{c}  {\tilde {\bf Y}}^{\mathtt{md}}_{\bf Q} \\  {\tilde {\bf X}}^{\mathtt{md}}_{\bf Q} \ea \Bigg)}^{-1} \Bigg\|_{\infty} \Bigg( \Bigg\|  \ba{cc}  {\tilde {\bf N}}^{\mathtt{md}} \\  {\tilde {\bf M}}^{\mathtt{md}}  \ea  ^{T}   \Bigg\|_\infty + \gamma \Bigg) \\
\leq & {\Bigg( \Bigg\| { I_p} + \sum \limits_{j=1}^{\infty} \bigg(\ba{cc}  -\Delta_{\bf \tilde{M}} &  -\Delta_{\bf \tilde{N}} \ea \ba{c}  {\tilde {\bf Y}}^{\mathtt{md}}_{\bf Q} \\  {\tilde {\bf X}}^{\mathtt{md}}_{\bf Q} \ea\bigg)^{j} \Bigg\|_{\infty}\Bigg)} \Bigg( \Bigg\|  \ba{cc}  {\tilde {\bf N}}^{\mathtt{md}} \\  {\tilde {\bf M}}^{\mathtt{md}}  \ea  ^{T}   \Bigg\|_\infty + \gamma \Bigg) \\
\leq & {\dfrac{1}{1 - \gamma \Bigg\| \ba{c}  {\tilde {\bf Y}}^{\mathtt{md}}_{\bf Q} \\  {\tilde {\bf X}}^{\mathtt{md}}_{\bf Q} \ea \Bigg\|_{\infty}} } \Bigg( \Bigg\|  \ba{cc}  {\tilde {\bf N}}^{\mathtt{md}} \\  {\tilde {\bf M}}^{\mathtt{md}}  \ea  ^{T}   \Bigg\|_\infty + \gamma \Bigg) 
\end{aligned}
\end{equation*}
Let,\\
$ \small f({\bf Q}) = \Bigg({1 - \gamma \Bigg\| \ba{@{}c@{}}  {\tilde {\bf Y}}^{\mathtt{md}}_{\bf Q} \\  {\tilde {\bf X}}^{\mathtt{md}}_{\bf Q} \ea \Bigg\|_{\infty}}\Bigg)  \Bigg\|    \ba{@{}c@{}} I_m - {\bf Y}_{\bf Q}^{\mathtt{md}} \\ {\bf X}_{\bf Q}^{\mathtt{md}} \ea  ^{T}   \Bigg\|_{\mathcal{H}_2} + \Big\| I_m - {{\bf M}^{\mathtt{md}}} \Big\|_\infty  \Big\| {\bf Q}   \Big\|_{\mathcal{H}_2} + \Bigg( \Bigg\|  \ba{@{}c@{}}  {\tilde {\bf N}}^{\mathtt{md}} \\  {\tilde {\bf M}}^{\mathtt{md}}  \ea  ^{T}   \Bigg\|_\infty + \gamma \Bigg)  $
\noindent
and $g({\bf Q}) = \gamma \Bigg\| \ba{c}  {\tilde {\bf Y}}^{\mathtt{md}}_{\bf Q} \\  {\tilde {\bf X}}^{\mathtt{md}}_{\bf Q} \ea \Bigg\|_{\infty} $.

\noindent
By applying \lemref{sarahopt1} and introducing one additional constraint $\Bigg\| \ba{c}  {\tilde {\bf Y}}^{\mathtt{md}}_{\bf Q} \\  {\tilde {\bf X}}^{\mathtt{md}}_{\bf Q} \ea \Bigg\|_{\infty} \leq \alpha$, the domain $\{ {\bf Q} \hspace{5pt} {\text {Stable}} \Big| \Bigg\| \ba{c}  {\tilde {\bf Y}}^{\mathtt{md}}_{\bf Q} \\  {\tilde {\bf X}}^{\mathtt{md}}_{\bf Q} \ea \Bigg\|_{\infty} \leq \dfrac{1}{\gamma} \}$ splits into two parts as $\{ {\bf Q} \hspace{5pt} {\text {Stable}} \Big| \Bigg\| \ba{c}  {\tilde {\bf Y}}^{\mathtt{md}}_{\bf Q} \\  {\tilde {\bf X}}^{\mathtt{md}}_{\bf Q} \ea \Bigg\|_{\infty} \leq \alpha \}$  and $\{ {\bf Q} \hspace{5pt} {\text {Stable}} \Big| \Bigg\| \ba{c}  {\tilde {\bf Y}}^{\mathtt{md}}_{\bf Q} \\  {\tilde {\bf X}}^{\mathtt{md}}_{\bf Q} \ea \Bigg\|_{\infty} \in (\alpha, \dfrac{1}{\gamma} ] \}$. \\

\noindent
 Thus,  the upper bound of \thmref{theoremAcost} is formulated into a Quasi-convex problem as:
\begin{equation*} 
\small
\begin{aligned}
        \min_{\alpha \in [0,1/\gamma)} & \dfrac{1}{1 - \gamma \alpha}  \min_{\bf Q\: \text{stable}}  (1-\gamma \alpha) \Big\| \ba{@{}c@{}} {\big( I_m - {\bf Y}_{\bf Q}^{\mathtt{md}}}\big) \\ {\bf X}_{\bf Q}^{\mathtt{md}} \ea ^T \Big\|_{\mathcal{H}_2} +   \Big\| \ba{@{}c@{}} I_m - {\bf M}^{\mathtt{md}} \ea  \Big\|_{\mathcal\infty}\big\| {\bf Q}  \big\|_{\mathcal{H}_2} \Big( \Big\| \ba{@{}c@{}} {\tilde {\bf N}}^{\mathtt{md}} \\ {\tilde {\bf M}}^{\mathtt{md}} \ea ^T \Big\|_{\mathcal\infty} + \gamma \Big)\\
        \centering
        &  \hspace{40pt}\textrm{s.t.} \quad
    \quad \Bigg\|
\ba{c}  {\tilde {\bf Y}_{\bf Q}^{\mathtt{md}}} \\  {\tilde {\bf X}}_{\bf Q}^{\mathtt{md}} \ea \Bigg\|_{\infty} \leq \alpha.
\end{aligned}
\end{equation*}

\end{proof}

\begin{prop}\label{SDPformulation2}(SDP formulation for \propref{theoremB})
For simplicity, the following is denoted:
\begin{equation*}
\begin{split}
    \ba{@{}c@{}} I_m - {\bf Y}^{\mathtt{md}}(z) \\ {\bf X}^{\mathtt{md}}(z) \ea = 
    \sum \limits_{j=0}^n C_j z^{-j}, \quad \quad 
    \ba{@{}c@{}} \tilde{\bf Y}^{\mathtt{md}}(z) \\ \tilde{\bf X}^{\mathtt{md}}(z) \ea = 
    \sum \limits_{j=0}^n H_j z^{-j}\\
    \ba{@{}c@{}} \tilde{\bf N}^{\mathtt{md}}(z) \\ \tilde{\bf M}^{\mathtt{md}}(z) \ea = 
    \sum \limits_{j=0}^n P_j z^{-j} \quad \quad
    \ba{@{}c@{}} - {\bf N}^{\mathtt{md}}(z) \\ \tilde{\bf M}^{\mathtt{md}}(z) \ea = 
    \sum \limits_{j=0}^n K_j z^{-j}.
\end{split}
\end{equation*}

\noindent
Note that all four time-domain representations above are assumed to be known. 
${\bf Y}^{\mathtt{md}}(z)$, ${\bf X}^{\mathtt{md}}(z)$, $\tilde{\bf Y}^{\mathtt{md}}(z)$, $\tilde{\bf X}^{\mathtt{md}}(z)$, $\tilde{\bf M}^{\mathtt{md}}(z)$, $\tilde{\bf N}^{\mathtt{md}}(z)$, ${\bf M}^{\mathtt{md}}(z)$, ${\bf N}^{\mathtt{md}}(z)$ are the DCFs of the initial controller and the nominal model. The system $I_m$ can be viewed as a system whose Markov parameters are given by $\delta[j]I_m$, where $\delta[j]$ is the Dirac Impulse function with respect to j.

\noindent
Next, the decision variable ${\bf Q}(z)$ is considered as ${\bf Q}(z) = \sum \limits_{j=0}^{n} Q_j z^{-j}$, define: 
\begin{equation*}
\hat{Q}_k = \ba{c:c:c:c} I_p \otimes Q_k^T & I_p \otimes Q_{k-1}^T & \dots & I_p \otimes Q_{k-n}^T \ea \quad
\hat{Q} =  \ba{c:c:c:c} \hat{Q}_0 & \hat{Q}_1 & \dots & \hat{Q}_n \ea ^T,
\end{equation*}

\noindent
where $\otimes$ denotes the Kronecker product, and by convention, it is considered that $Q_{-1} = Q_{-2} = \dots = 0$.

\noindent
The Markov parameters of the system $\ba{c} \tilde{\bf Y}^{\mathtt{md}}_{\bf Q}(z) \\ \tilde{\bf X}^{\mathtt{md}}_{\bf Q}(z) \ea$ can be written in the sense of convolution as:
\begin{equation*}
    \ba{c} \tilde{\bf Y}^{\mathtt{md}}_{\bf Q}(z) \\ \tilde{\bf X}^{\mathtt{md}}_{\bf Q}(z) \ea =  \sum \limits_{j=0}^n F_j (Q) z^{-j}, \quad \quad  F_j (Q) = H_j + \sum \limits_k K_{(j-k)} Q_k.
\end{equation*}
\noindent
The first $(n+1)$ Markov parameters are taken after the convolution and denoted as:\\ $\mathcal{F}(Q) = \ba{cccc} {F}_0(Q) & {F}_1(Q) & \dots & {F}_n(Q) \ea ^T \in \mathbb{R}^{[(p+m)\times n] \times p}$.

\noindent
Now, denote $\Big\| \ba{c} I_m - {\bf M}^{\mathtt{md}} \ea \Big\|_\infty  = \lambda_1$ and $\Bigg\| \ba{c} {\bf N^{\mathtt{md}}} \\ {\bf M^{\mathtt{md}}} \ea \Bigg\|_\infty = \lambda_2$.
Let $\mathtt{vec(.)}$ be the column vectorization of a matrix. Also let,
\begin{equation*}
\begin{split}
\overline{P} = \ba{cccc} \mathtt{vec}({P}_0) & \mathtt{vec}({P}_1) & \dots & \mathtt{vec}({P}_n) \ea^T, \quad 
\overline{C} = \ba{cccc} \mathtt{vec}(C_0) & \mathtt{vec}(C_1) & \dots & \mathtt{vec}(C_n) \ea^T \\
and \hspace{10pt} \overline{Q} = \ba{cccc} \mathtt{vec}(Q_0) & \mathtt{vec}(Q_1) & \dots & \mathtt{vec}(Q_n) \ea^T.
\hspace{195 pt} 
\end{split}
\end{equation*}

\noindent
By the fact that the constraint $\Bigg\|
\ba{c}  {\tilde {\bf Y}}_{\bf Q}^{\mathtt{md}} \\  {\tilde {\bf X}}_{\bf Q}^{\mathtt{md}} \ea \Bigg\|_{\infty} \leq \alpha$ is equivalent with the existence of a semi-positive definite matrix $S$ whole block structure is given by $\sum \limits_{i=1}^{n-k} S_{i+k,i} = \alpha \delta[k] I_{p+m}, k = 0:n$. Then the following is stated:

\noindent
The inner optimization problem in \propref{theoremB} after FIR truncation up to $(n+1)^{th}$ Markov parameter can be formulated as:
\begin{equation} \label{FIRtruncation2}
\begin{aligned}
  \min_{S, \overline{Q}, \epsilon_1, \epsilon_2}  & (1-\gamma \alpha) \epsilon_1 + \lambda_1 (\lambda_2+ \gamma) \epsilon_2  \\
\textrm{s.t.} \quad &  S \in \mathcal{S}_+^{(p+m)(n+1)}, \quad \sum \limits_{i=1}^{n-k} S_{i+k, i} = \alpha \delta[k] I_{p+m},\quad k = 0:n,\\
& \ba{cc} S & \mathcal{F}(Q)\\ \mathcal{F}^T (Q) & \alpha I_p \ea \succeq 0,\\
& \|\overline{C} + \hat{Q}\overline{P} \| \leq \epsilon_1,\\
& \|\overline{Q}  \| \leq \epsilon_2.
\end{aligned}
\end{equation}
\end{prop}

\section{Optimal Cost}\label{appendixC}
If the pair $(A, B)$ is controlable and the pair $(A,C)$ is observable, then it is appropriate to introduce the time domain-representation of the $\mathcal{H}_2$ cost as:
\begin{equation} \label{h2costforobserver}
    J = \sum \limits_{k=1}^{m} \{ \sum \limits_{t=0}^{\infty} y^T_t y_t, \quad w_t = e_k \delta_t, \quad \nu_k = 0  \}
\end{equation}

\noindent
Here in \eqref{h2costforobserver}, $e_k$ represents the $k^{\text{th}}$ standard basis vector in $\mathbb{R}^{m }$ and $\delta_t = \begin{cases}
			1, & t = 0\\
            0, & \text{otherwise}
		 \end{cases}$.
   
\noindent
The direct feedthrough from $\nu$ to $y$ is assumed to be none in order to obtain a finite $\mathcal{H}_2$ norm for the closed loop system. This characterization of $\mathcal{H}_2$ norm is not usually seen since the motivation of $\mathcal{H}_2$ optimal problem is more naturally stated by average frequency domain characterization, but it would explain the role of difference in control signals here.

\noindent
The optimal state feedback gain is denoted as $F^{\text{opt}}$. From the well-known Riccati  theory, $F^{\text{opt}} = - (D^T D)^{-1} B^T S$ with $S$ being the unique symmetric semidefinite solution to the Algebraic Riccati Equation (ARE):
\begin{equation} \label{Riccati}
    A^T S + S A - S B (D^T D)^{-1} B^T S + C^T C = 0
\end{equation}
\noindent
Furthermore, the matrix $A+BF^{\text{opt}}$ is stable.

Denote the optimal $\mathcal{H}_2$ control signal as $u^{\text{opt}}_t =  F^{\text{opt}} {x}_t$, then it's possible to rewrite the $\mathcal{H}_2$ cost as the following when applying a certain control $u_t$:
\begin{equation*}
    \begin{aligned}
        J_{u_t} & = \sum \limits_{k=1}^m \Big\{ \sum \limits_{t=0}^\infty [ (u_t -   u^{\text{opt}}_t )^T D^T D (u_t -   u^{\text{opt}}_t ) ]  + e_k^T B^T S B  e_k ; \quad w_t= e_k \delta_t   \Big\}\\
        & = \sum \limits_{k=1}^m \{ \sum \limits_{t=0}^\infty [ (u_t -   u^{\text{opt}}_t )^T D^T D (u_t -   u^{\text{opt}}_t ) ]  ; \quad w_t = e_k \delta_t    \} + tr( B^T S B)
    \end{aligned}
\end{equation*}

\noindent
In our settings, the control signal used is $\hat{u}_t = F \hat{x}$, where $\hat{x}$ is generated by the designed observer such that it can be written as following:

\begin{equation*}
        J_{\hat{u}_t} = \sum \limits_{k=1}^m \{ \sum \limits_{t=0}^\infty [ ( \hat{u}_t -   u^{\text{opt}}_t )^T D^T D (\hat{u}_t -   u^{\text{opt}}_t ) ]  ; \quad w_t = e_k \delta_t    \} + tr( B^T S B) 
\end{equation*}

\noindent
Without loss of generality, it can be assumed that $D^T D = I$ for convenience. It can be observed that the term $((\hat{u}_t -   u^{\text{opt}}_t ))^T ((\hat{u}_t -   u^{\text{opt}}_t ))$ is non-negative and follows triangle inequality such that:

\begin{equation} \label{error_in_cost}
\begin{split}
    J_{\hat{u}_t} \leq \sum \limits_{k=1}^m \{ \sum \limits_{t=0}^\infty [ ( F{x}_t - F  \hat{x}_t )^T ( F{x}_t - F  \hat{x}_t ) ]  ;\quad w_t = e_k \delta_t    \} \\ + \sum \limits_{k=1}^m \{ \sum \limits_{t=0}^\infty [ (  u^{\text{opt}}_t - F{x}_t )^T ( u^{\text{opt}}_t - F{x}_t) ]  ; \quad w_t = e_k \delta_t    \} + tr( B^T S B)
    \end{split}
\end{equation}

\noindent
The $\mathcal{H}_2$ cost of applying the control $u_t = F \hat{x}_t$ is upperbounded as the above. Note that the third term is definitive and the second term is fixed as we are not able to change $F$. Then the only thing we seek to minimize is the first term, equivalently $\| F x(z) - F \hat{x} (z) \|^2_2$.

\noindent
A possible future topic for this part is that it is reasonable to reduce $\| u^{\text{opt}}_t - F{x}_t \|_2^2$ to compress the upperbound of $J_{\hat{u}_t}$ by learning the optimal feedback gain $F^{\text{opt}}$. To do so, a precise identification of system parameters $A, B, C, D$ is necessary. There already exists such algroithms as in \cite{papas2019} and \cite{Sarkar2019}.\\
Note that if the optimal state feedback gain $F^{\text {opt}}$ is given, then the second term in the right hand side of inequality\eqref{error_in_cost} is gone. In this case, by expressing the optimal $\mathcal{H}_2$-cost as $tr( B^T S B)$, our objective function in \thmref{theoremAcost} is directly the error in $\mathcal{H}_2$-cost comparing to the optimal LQ regulator for the true plant. Solving \thmref{theoremAcost} yields a robust stabilizing controller such that its performance is guaranteed by \eqref{costJujuhat}. This provides a new perspective for robust controller design via an observer approach, parallel to our previous work \cite{Zheng2022}.



\section{Closed Loop Identification Scheme} \label{appendixD}

 Details on the closed-loop identification scheme of a {\em noise contaminated plant} $\textbf{G}^\mathtt{md}$ with control input $u$, noise $\nu$ (taken $w$ = $0$) and output measurement $y$ (where  $u$ and  $\nu$ are assumed independent and stationary) is depicted on Figure~2 below.

 \begin{figure}[h]
\centering
\includegraphics[width=15.6cm]{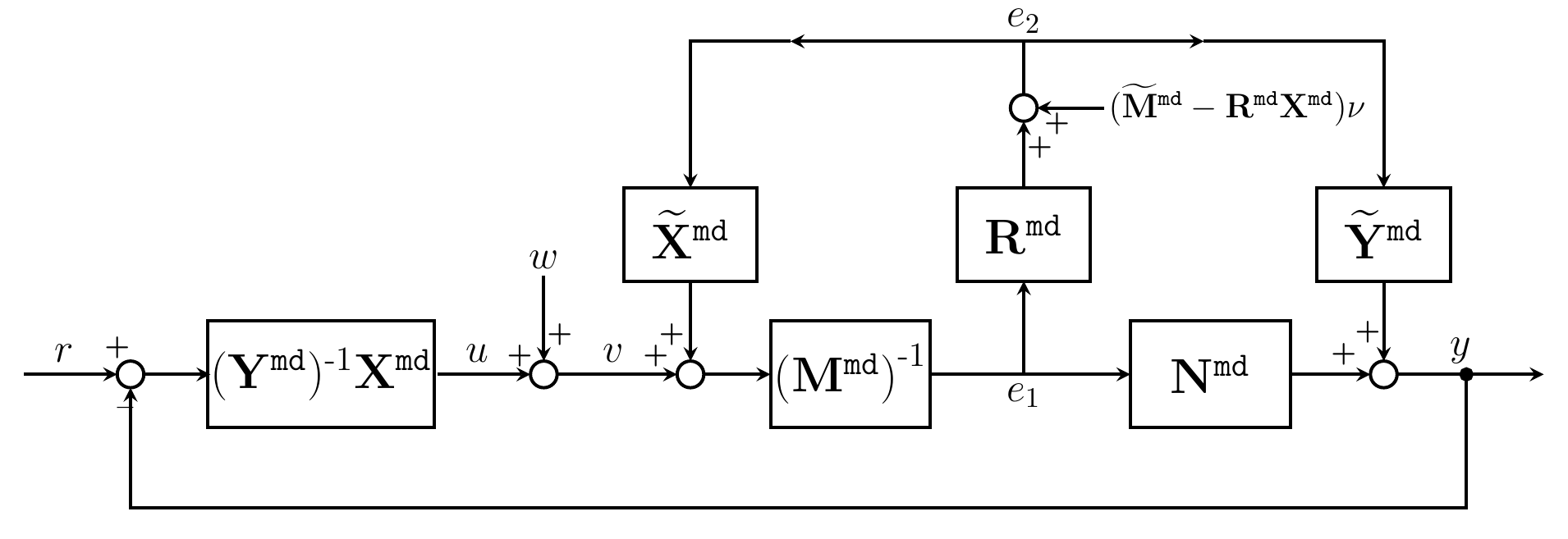}
\caption{Closed loop Identification for Noise Contaminated Plant}
\label{clsysid2}
\end{figure}

 \noindent
The main idea dating back to \cite{Anderson1998} is to identify the {\em stable} dual-Youla parameter ${\bf R}^{\mathtt{md}}$ rather than ${\bf G}^{\mathtt{md}}$ thus recasting the problem in a standard, open-loop identification form. In this way, model uncertainties are additive to the coprime factors of model, not directly on the model. For details on this and identification algorithms, we refer to Appendix G of \cite{Zheng2022arxiv}.

\end{document}